\documentclass[11pt,letterpaper]{article}
\usepackage{color,ifpdf,latexsym,graphicx,url,amsthm,amsmath,subcaption,float}

\usepackage{thmtools,thm-restate}

\usepackage{hyperref}

\advance \topmargin by -\headheight
\advance \topmargin by -\headsep
\evensidemargin \oddsidemargin
{\makeatletter \gdef\fps@figure{!htbp}}

\let\realbfseries=\bfseries
\def\bfseries{\realbfseries\boldmath}

\newtheorem{theorem}{Theorem}[section]
\newtheorem{lemma}[theorem]{Lemma}
\newtheorem{corollary}[theorem]{Corollary}
\newtheorem{definition}[theorem]{Definition}
\newtheorem{problem}{Problem}

{\makeatletter
 \gdef\xxxmark{%
   \expandafter\ifx\csname @mpargs\endcsname\relax 
     \expandafter\ifx\csname @captype\endcsname\relax 
       \marginpar{xxx}
     \else
       xxx 
     \fi
   \else
     xxx 
   \fi}
 \gdef\xxx{\@ifnextchar[\xxx@lab\xxx@nolab}
 \long\gdef\xxx@lab[#1]#2{\textbf{[\xxxmark #2 ---{\sc #1}]}}
 \long\gdef\xxx@nolab#1{\textbf{[\xxxmark #1]}}
}

\makeatletter
\def\GrabProofArgument[#1]{ #1: \egroup\ignorespaces}
\def\proof{\noindent\textbf\bgroup Proof%
           \@ifnextchar[{\GrabProofArgument}{: \egroup\ignorespaces}}

\makeatother

\begin{document}

\title{Solving the Rubik's Cube Optimally is NP-complete}
\author{
Erik D. Demaine%
    \thanks{MIT Computer Science and Artificial Intelligence Laboratory,
      32 Vassar St., Cambridge, MA 02139, USA,
      \protect\url{edemaine@mit.edu}}
\and
  Sarah Eisenstat\footnotemark[1]
\and
  Mikhail Rudoy%
  \thanks{MIT Computer Science and Artificial Intelligence Laboratory,
      32 Vassar St., Cambridge, MA 02139, USA,
      \protect\url{mrudoy@gmail.com}. Now at Google Inc.}
}
\date{}

\maketitle

\begin{abstract}
In this paper, we prove that optimally solving an $n \times n \times n$ Rubik's Cube is NP-complete by reducing from the Hamiltonian Cycle problem in square grid graphs. This improves the previous result that optimally solving an $n \times n \times n$ Rubik's Cube with missing stickers is NP-complete. We prove this result first for the simpler case of the Rubik's Square---an $n \times n \times 1$ generalization of the Rubik's Cube---and then proceed with a similar but more complicated proof for the Rubik's Cube case. Our results hold both when the goal is make the sides monochromatic and when the goal is to put each sticker into a specific location.
\end{abstract}

\section{Introduction}

The Rubik's Cube is an iconic puzzle in which the goal is to rearrange the stickers on the outside of a $3 \times 3 \times 3$ cube so as to make each face monochromatic by rotating $1 \times 3 \times 3$ (or $3 \times 1 \times 3$ or $3 \times 3 \times 1$) slices. In some versions where the faces show pictures instead of colors, the goal is to put each sticker into a specific location. The $3 \times 3 \times 3$ Rubik's Cube can be generalized to an $n \times n \times n$ cube in which a single move is a rotation of a $1 \times n \times n$ slice. We can also consider the generalization to an $n \times n \times 1$ figure. In this simpler puzzle, called the $n \times n$ Rubik's Square, the allowed moves are flips of $n \times 1 \times 1$ rows or $1 \times n \times 1$ columns. These two generalizations were introduced in \cite{demaine}.

The overall purpose of this paper is to address the computational difficulty of optimally solving these puzzles. In particular, consider the decision problem which asks for a given puzzle configuration whether that puzzle can be solved in a given number of moves. We show that this problem is NP-complete for the $n \times n$ Rubik's Square and for the $n \times n \times n$ Rubik's Cube under two different move models. These results close a problem that has been repeatedly posed as far back as 1984 \cite{cook, ratner, stackexchange} and has until now remained open \cite{kendall}.

In Section~\ref{section:puzzle_problems}, we formally introduce the decision problems regarding Rubik's Squares and Rubik's Cubes whose complexity we will analyze. Then in Section~\ref{section:promise_problems}, we introduce the variant of the Hamiltonicity problem that we will reduce from---Promise Cubical Hamiltonian Path---and prove this problem to be NP-hard. Next, we prove that the problems regarding the Rubik's Square are NP-complete in Section~\ref{section:rubiks_square} by reducing from Promise Cubical Hamiltonian Path. After that, we apply the same ideas in Section~\ref{section:rubiks_cube} to a more complicated proof of NP-hardness for the problems regarding the Rubik's Cube. Finally, we discuss possible next steps in Section~\ref{section:next_steps}.

\section{Rubik's Cube and Rubik's Square problems}
\label{section:puzzle_problems}

\subsection{Rubik's Square}

We begin with a simpler model based on the Rubik's Cube which we will refer to as the Rubik's Square. In this model, a puzzle consists of an $n \times n$ array of unit cubes, called \emph{cubies} to avoid ambiguity. Every cubie face on the outside of the puzzle has a colored (red, blue, green, white, yellow, or orange) sticker. The goal of the puzzle is to use a sequence of moves to rearrange the cubies such that each face of the puzzle is monochromatic in a different color. A \emph{move} consists of flipping a single row or column in the array through space via a rotation in the long direction as demonstrated in Figure~\ref{fig:square_move}. 

\begin{figure}[h]
\centering
\includegraphics[width=.15\textwidth]{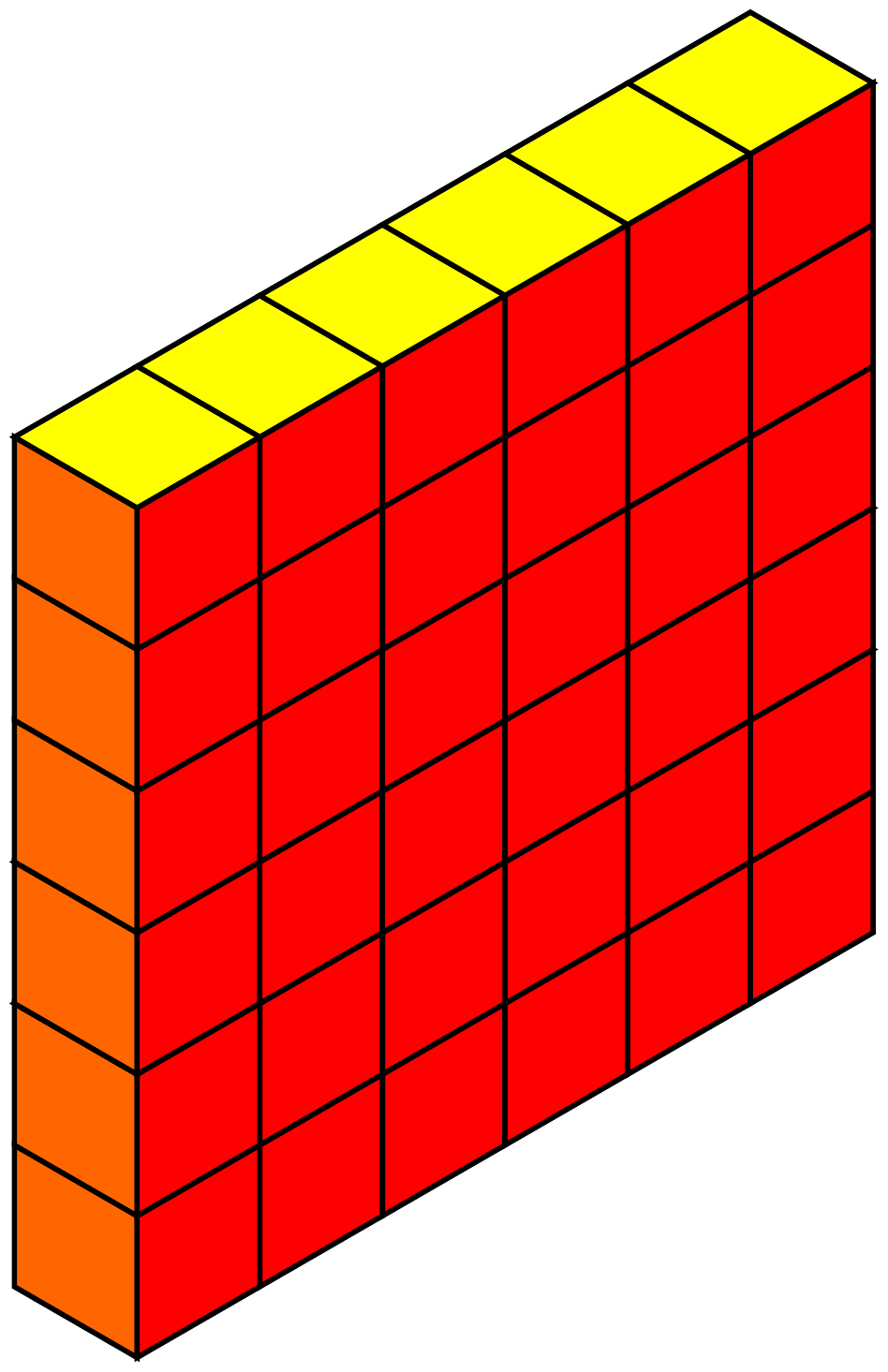}
\hfill
\includegraphics[width=.15\textwidth]{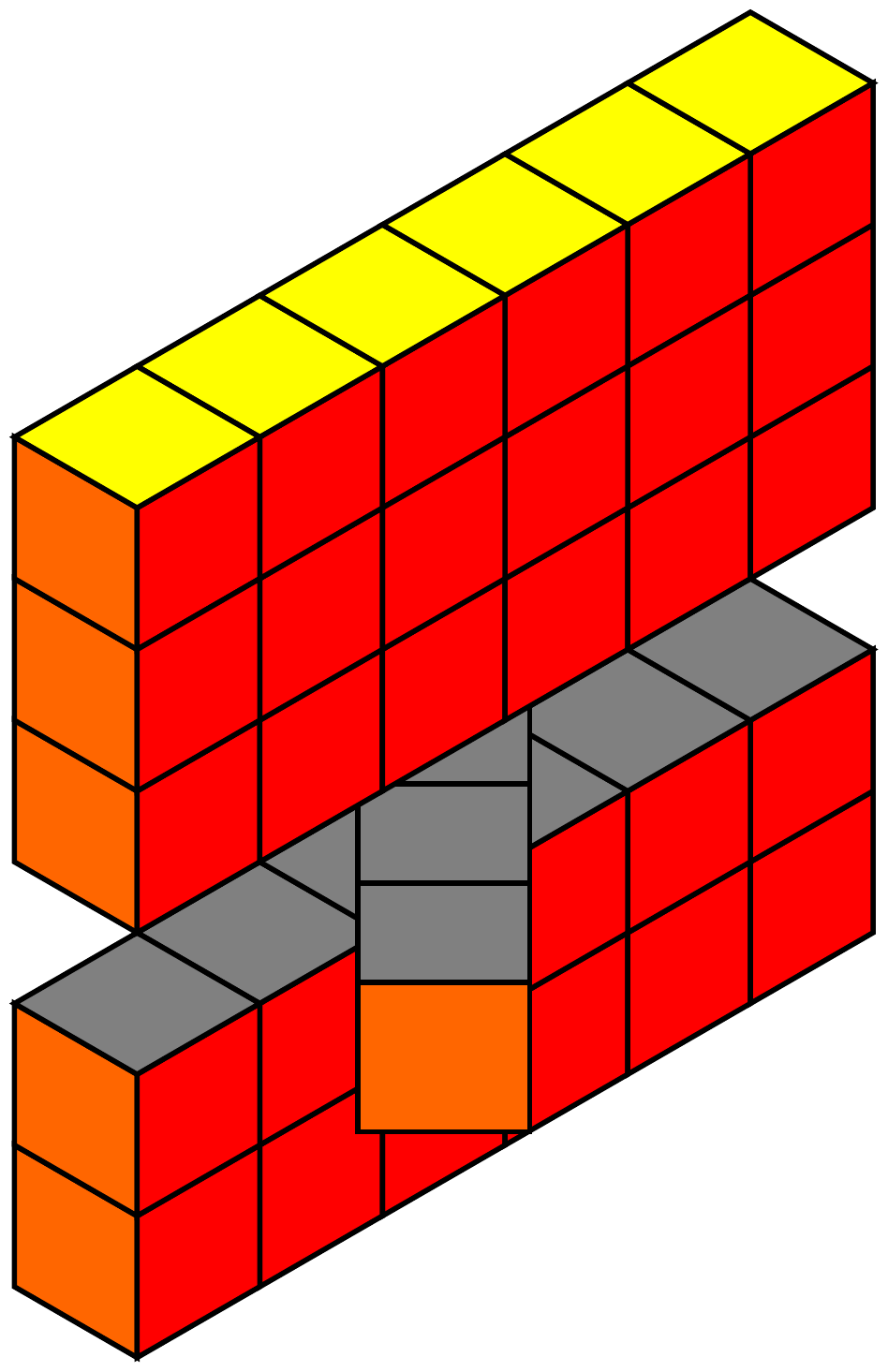}
\hfill
\includegraphics[width=.15\textwidth]{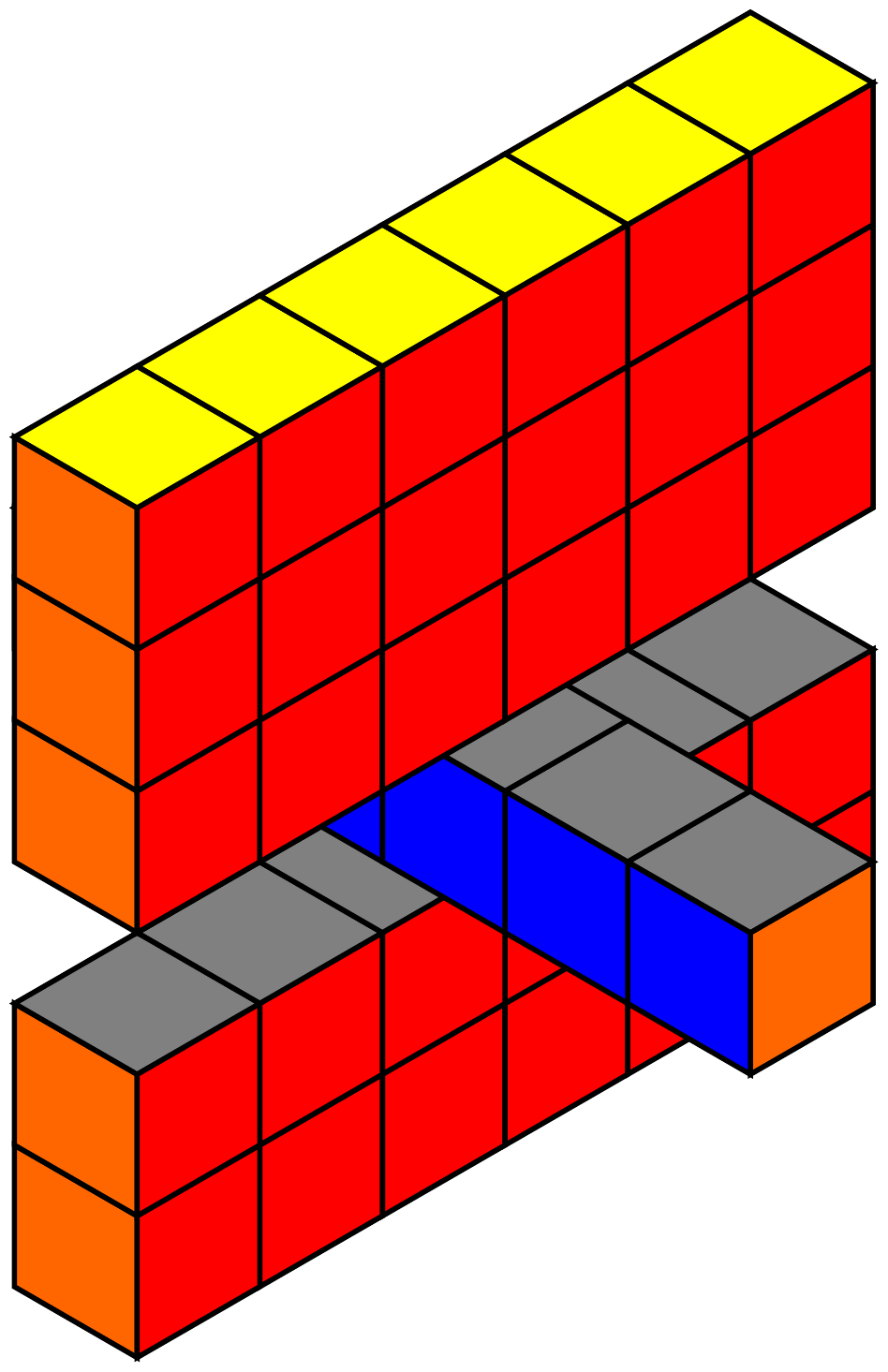}
\hfill
\includegraphics[width=.18\textwidth]{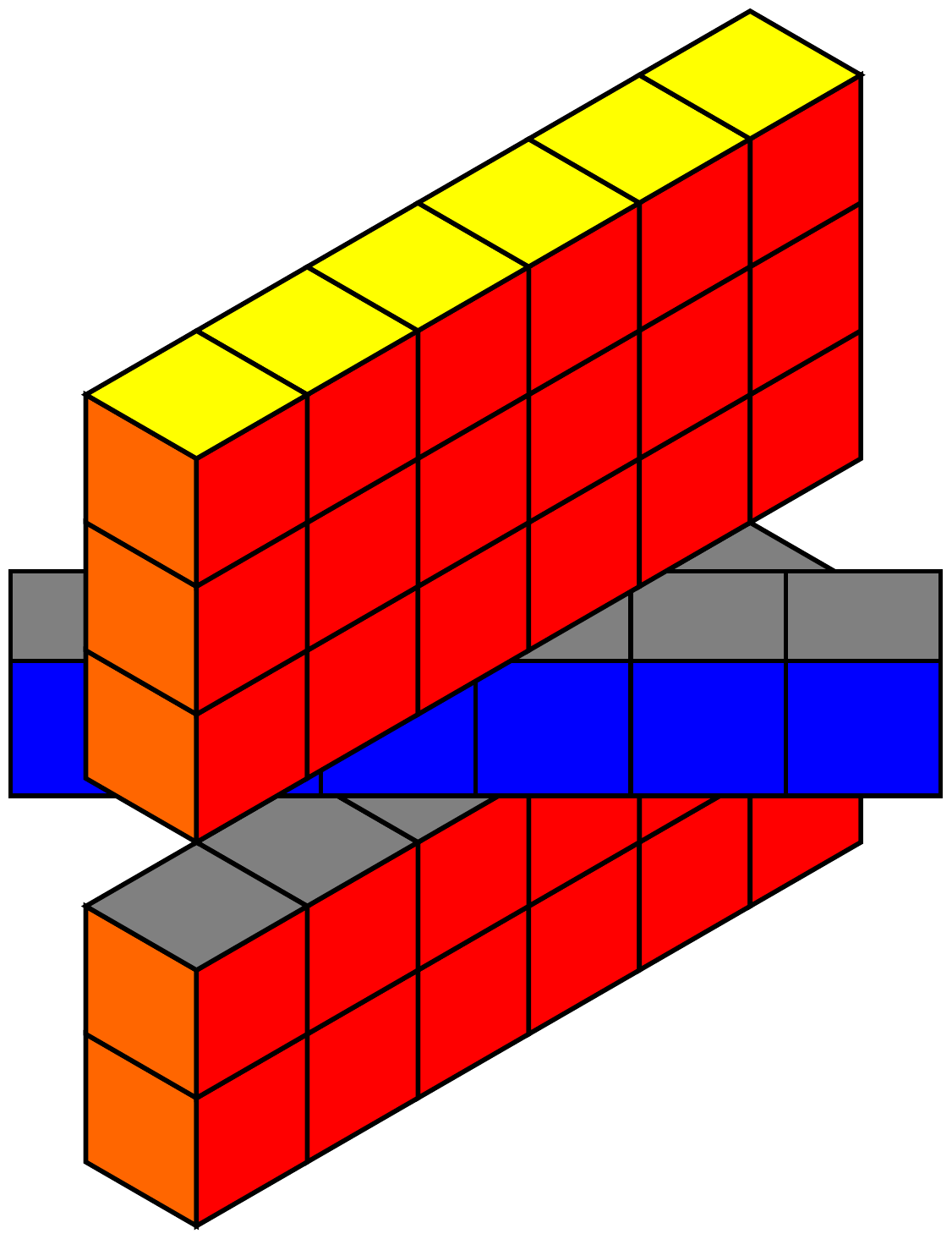}
\hfill
\includegraphics[width=.15\textwidth]{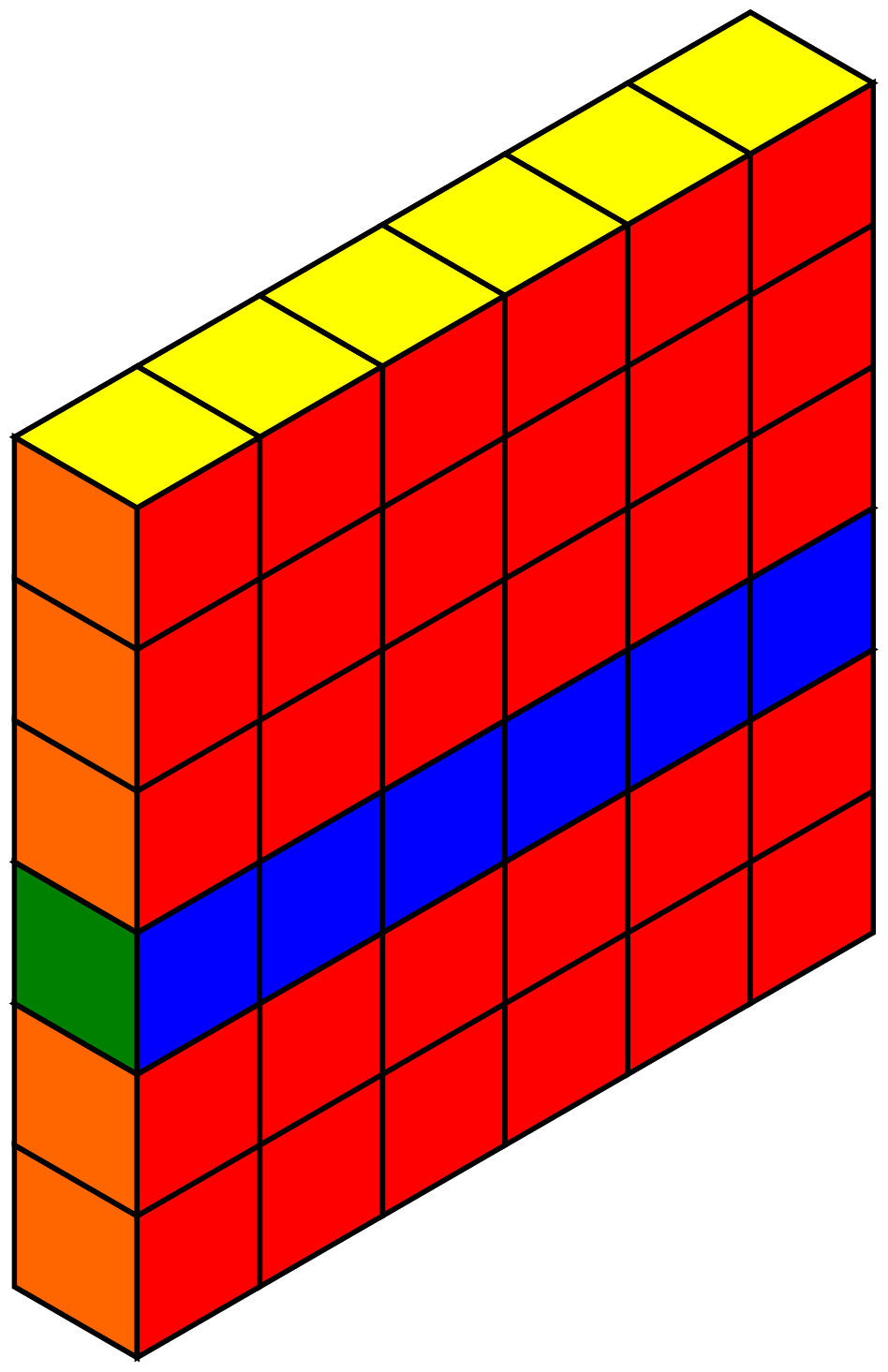}
\caption{A single move in an example $6 \times 6$ Rubik's Square.}
\label{fig:square_move}
\end{figure}

We are concerned with the following decision problem:

\begin{problem}
The \textbf{Rubik's Square} problem has as input an $n \times n$ Rubik's Square configuration and a value $k$. The goal is to decide whether a Rubik's Square in configuration $C$ can be solved in $k$ moves or fewer.
\end{problem}

Note that this type of puzzle was previously introduced in \cite{demaine} as the $n \times n \times 1$ Rubik's Cube. In that paper, the authors showed that deciding whether it is possible to solve the $n \times n \times 1$ Rubik's Cube in a given number of moves is NP-complete when the puzzle is allowed to have missing stickers (and the puzzle is considered solved if each face contains stickers of only one color). 

\subsection{Rubik's Cube}

Next consider the Rubik's Cube puzzle. An $n \times n \times n$ Rubik's Cube is a cube consisting of $n^3$ unit cubes called \emph{cubies}. Every face of a cubie that is on the exterior of the cube has a colored (red, blue, green, white, yellow, or orange) sticker. The goal of the puzzle is to use a sequence of moves to reconfigure the cubies in such a way that each face of the cube ends up monochromatic in a different color. A \emph{move count metric} is a convention for counting moves in a Rubik's Cube. Several common move count metrics for Rubik's Cubes are listed in \cite{wiki}. As discussed in \cite{forum}, however, many common move count metrics do not easily generalize to $n > 3$ or are not of any theoretical interest. In this paper, we will restrict our attention to two move count metrics called the Slice Turn Metric and the Slice Quarter Turn Metric. Both of these metrics use the same type of motion to define a move. Consider the subdivision of the Rubik's Cube's volume into $n$ \emph{slices} of dimension $1 \times n \times n$ (or $n \times 1 \times n$ or $n \times n \times 1$). In the Slice Turn Metric (STM), a \emph{move} is a rotation of a single slice by any multiple of $90^\circ$. Similarly, in the Slice Quarter Turn Metric (SQTM), a \emph{move} is a rotation of a single slice by an angle of $90^\circ$ in either direction. An example SQTM move is shown in Figure~\ref{fig:cube_move}.

\begin{figure}[h]
\centering
\includegraphics[width=.3\textwidth]{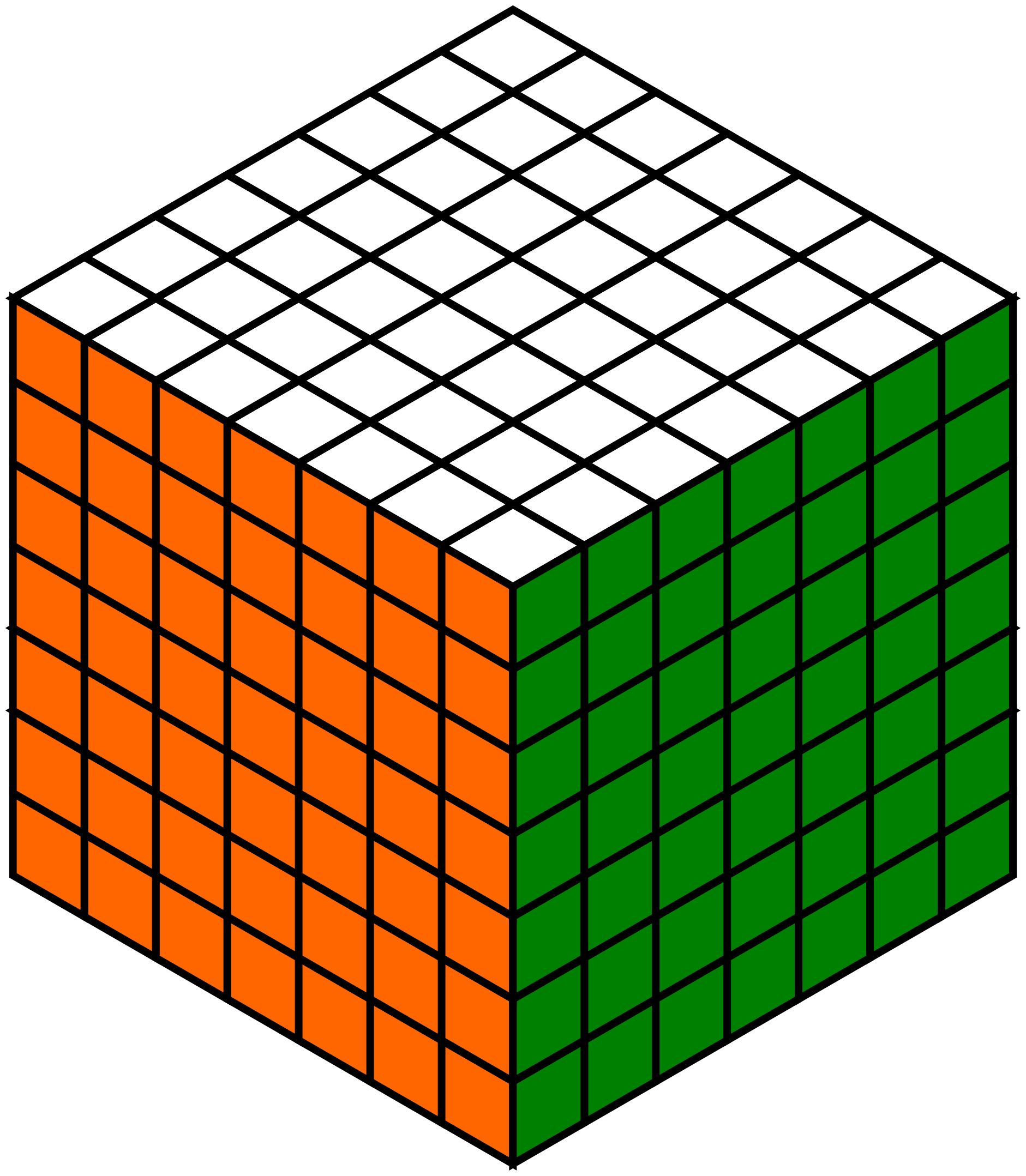}
\hfill
\includegraphics[width=.3\textwidth]{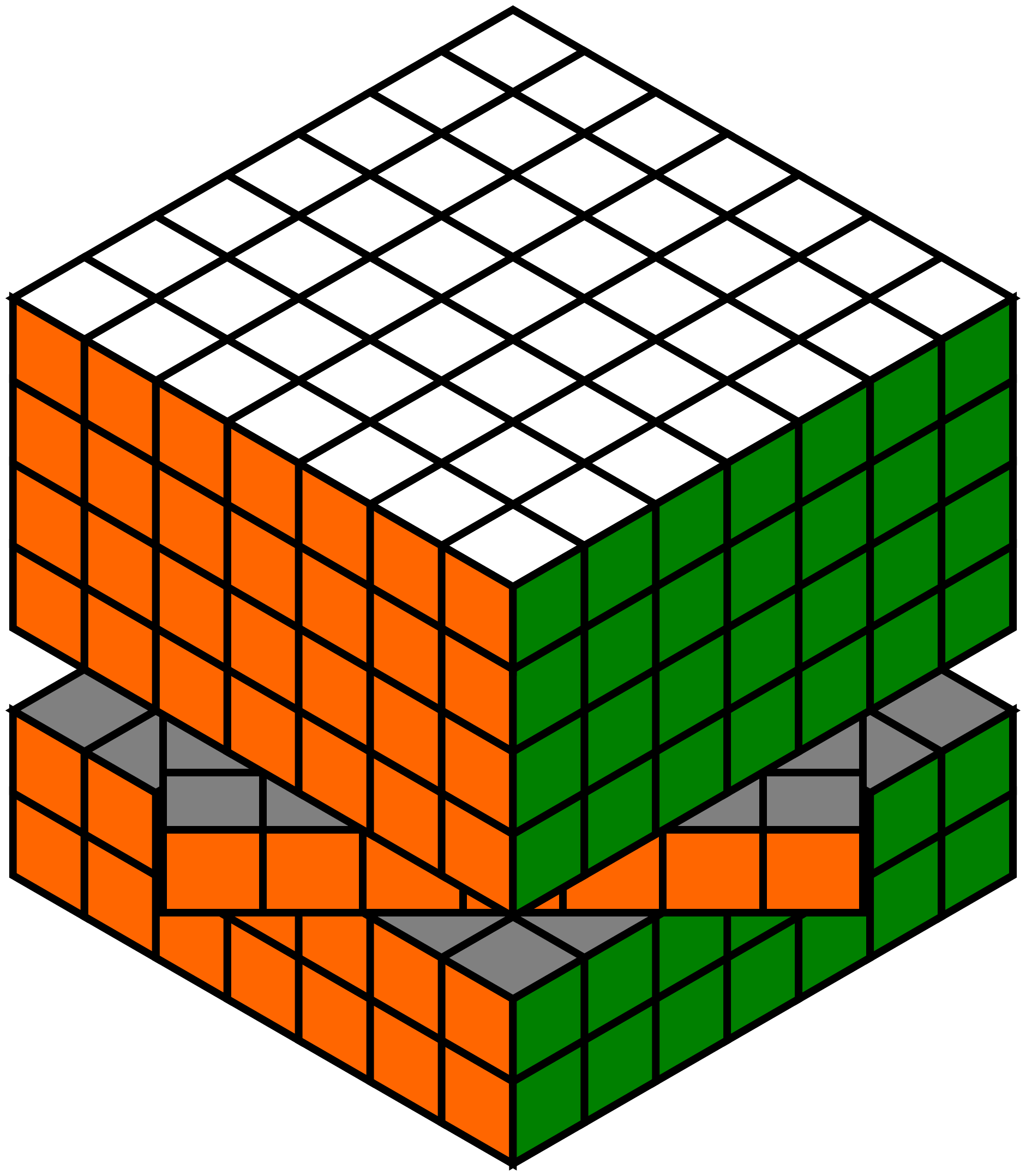}
\hfill
\includegraphics[width=.3\textwidth]{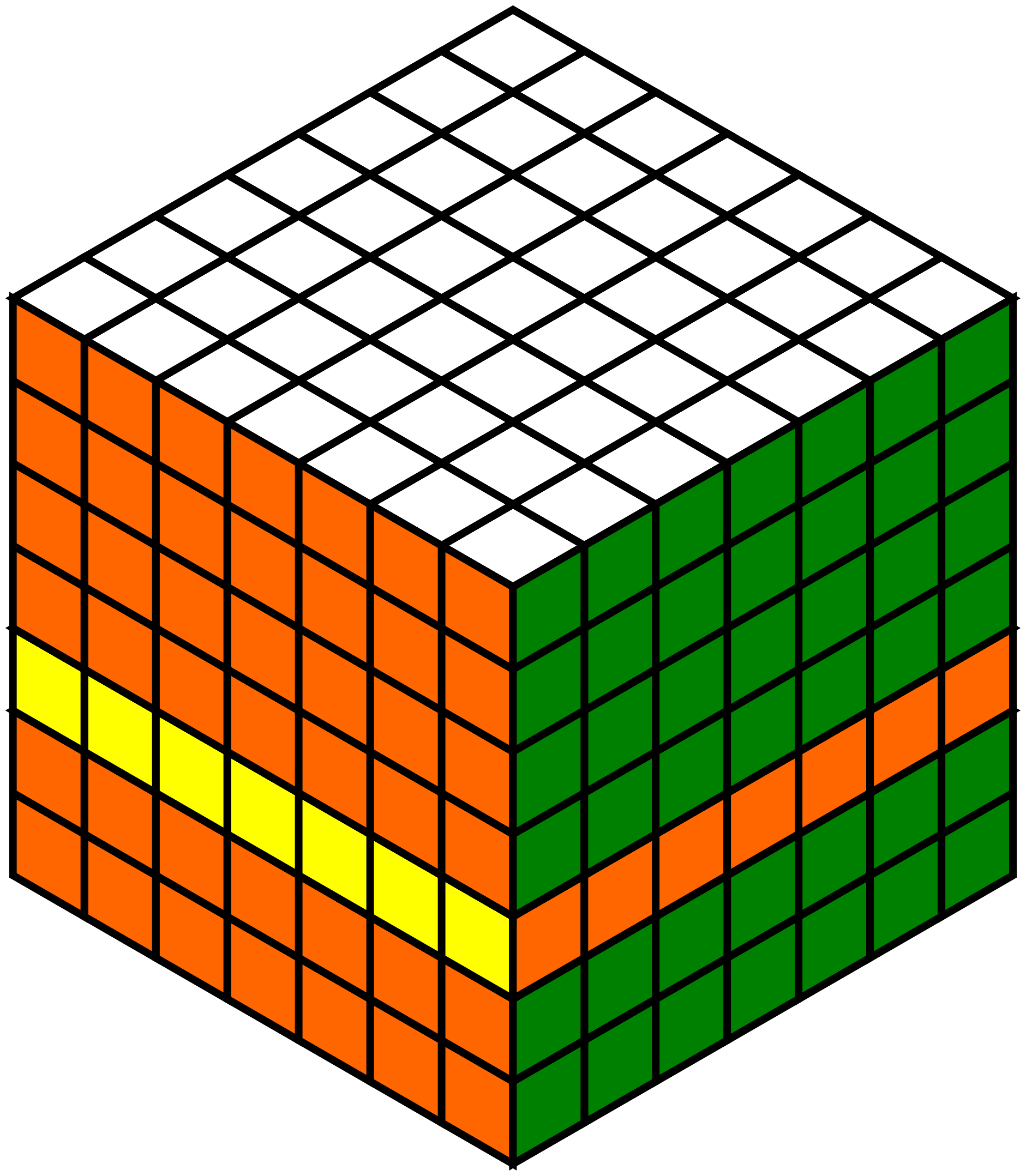}
\caption{A single slice rotation in an example $7 \times 7 \times 7$ Rubik's Cube.}
\label{fig:cube_move}
\end{figure}

We are concerned with the following decision problems:

\begin{problem}
The \textbf{STM/SQTM Rubik's Cube problem} takes as input a configuration of a Rubik's Cube together with a number $k$. The goal is to decide whether a Rubik's Cube in configuration $C$ can be solved in at most $k$ STM/SQTM moves.
\end{problem}

\subsection{Notation}

Next we define some notation for dealing with the Rubik's Cube and Rubik's Square problems.

To begin, we need a way to refer to cubies and stickers. For this purpose, we orient the puzzle to be axis-aligned. In the case of the Rubik's Square we arrange the $n \times n$ array of cubies in the $x$ and $y$ directions and we refer to a cubie by stating its $x$ and $y$ coordinates. In the case of the Rubik's Cube, we refer to a cubie by stating its $x$, $y$, and $z$ coordinates. To refer to a sticker in either puzzle, we need only specify the face on which that sticker resides (e.g. ``top'' or ``$+z$'') and also the two coordinates of the sticker along the surface of the face (e.g. the $x$ and $y$ coordinates for a sticker on the $+z$ face).

If $n = 2a+1$ is odd, then we will let the coordinates of the cubies in each direction range over the set $\{-a, -(a-1), \ldots, -1, 0, 1, \ldots, a-1, a\}$. This is equivalent to centering the puzzle at the origin. If, however, $n = 2a$ is even, then we let the coordinates of the cubies in each direction range over the set $\{-a, -(a-1), \ldots, -1\} \cup \{1, \ldots, a-1, a\}$. In this case, the coordinate scheme does not correspond with a standard coordinate sheme no matter how we translate the cube. This coordinate scheme is a good idea for the following reason: under this scheme, if a move relocates a sticker, the coordinates of that sticker remain the same up to permutation and negation.

Next, we need a way to distinguish the sets of cubies affected by a move from each other. 

In the Rubik's Square, there are two types of moves. The first type of move, which we will call a \emph{row move} or a \emph{$y$ move}, affects all the cubies with some particular $y$ coordinate. The second type of move, which we will call a \emph{column move} or an \emph{$x$ move} affects all the cubies with some particular $x$ coordinate. We will refer to the set of cubies affected by a row move as a \emph{row} and refer to the set of cubies affected by a column move as a \emph{column}. In order to identify a move, we must identify which row or column is being flipped, by specifying whether the move is a row or column move as well as the index of the coordinate shared by all the moved cubies (e.g. the index $-5$ row move is the move that affects the cubies with $y = -5$).

In the Rubik's Cube, each STM/SQTM move affects a single slice of $n^2$ cubies sharing some coordinate. If the cubies share an $x$ (or $y$ or $z$) coordinate, then we call the slice an \emph{$x$ (or $y$ or $z$) slice}. As with the Rubik's Square, we identify the slice by its normal direction together with its cubies' index in that direction (e.g. the $x = 3$ slice). We will also refer to the six slices at the boundaries of the Cube as \emph{face slices} (e.g. the $+x$ face slice). 

A move in a Rubik's Cube can be named by identifying the slice being rotated and the amount of rotation. We split this up into the following five pieces of information: the normal direction to the slice, the sign of the index of the slice, the absolute value of the index of the slice, the amount of rotation, and the direction of rotation. Splitting the information up in this way allows us not only to refer to individual moves (by specifying all five pieces of information) but also to refer to interesting sets of moves (by omitting one or more of the pieces of information).

To identify the normal direction to a slice, we simply specify $x$, $y$, or $z$; for example, we could refer to a move as an $x$ move whenever the rotating slice is normal to the $x$ direction. We will use two methods to identify the sign of the index of a moved slice. Sometimes we will refer to positive moves or negative moves, and sometimes we will combine this information with the normal direction and specify that the move is a $+x$, $-x$, $+y$, $-y$, $+z$, or $-z$ move. We use the term \emph{index-$v$ move} to refer to a move rotating a slice whose index has absolute value $v$. In the particular case that the slice rotated is a face slice, we instead use the term \emph{face move.} We refer to a move as a \emph{turn} if the angle of rotation is $90^\circ$ and as a \emph{flip} if the angle of rotation is $180^\circ$. In the case that the angle of rotation is $90^\circ$, we can specify further by using the terms \emph{clockwise turn} and \emph{counterclockwise turn}. We make the notational convention that clockwise and counterclockwise rotations around the $x$, $y$, or $z$ axes are labeled according to the direction of rotation when looking from the direction of positive $x$, $y$, or $z$.

We also extend the same naming conventions to the Rubik's Square moves. For example, a positive row move is any row move with positive index and an index-$v$ move is any move with index $\pm v$.

\subsection{Group-theoretic approach}

An alternative way to look at the Rubik's Square and Rubik's Cube problems is through the lens of group theory. The transformations that can be applied to a Rubik's Square or Rubik's Cube by a sequence of moves form a group with composition as the group operation. Define $RS_n$ to be the group of possible sticker permutations in an $n \times n$ Rubik's Square and define $RC_n$ to be the group of possible sticker permutations in an $n \times n \times n$ Rubik's Cube. 

Consider the moves possible in an $n\times n$ Rubik's Square or an $n \times n \times n$ Rubik's Cube. Each such move has a corresponding element in group $RS_n$ or $RC_n$. 

For the Rubik's Square, let $x_i \in RS_n$ be the transformation of flipping the column with index $i$ in an $n \times n$ Rubik's Square and let $y_i$ be the transformation of flipping the row with index $i$ in the Square. Then if $I$ is the set of row/column indices in an $n\times n$ Rubik's Square we have that $RS_n$ is generated by the set of group elements $\bigcup_{i\in I}\{x_i, y_i\}$.

Similarly, for the Rubik's Cube, let $x_i$, $y_i$, and $z_i$ in $RC_n$ be the transformations corresponding to clockwise turns of $x$, $y$, or $z$ slices with index $i$. Then if $I$ is the set of slice indices in an $n\times n \times n$ Rubik's Cube we have that $RC_n$ is generated by the set of group elements $\bigcup_{i\in I}\{x_i, y_i, z_i\}$. 

Using these groups we obtain a new way of identifying puzzle configurations. Let $C_0$ be a canonical solved configuration of a Rubik's Square or Rubik's Cube puzzle. For the $n \times n$ Rubik's Square, define $C_0$ to have top face red, bottom face blue, and the other four faces green, orange, yellow, and white in some fixed order. For the $n \times n \times n$ Rubik's Cube, let $C_0$ have the following face colors: the $+x$ face is orange, the $-x$ face is red, the $+y$ face is green, the $-y$ face is yellow, the $+z$ face is white, and the $-z$ face is blue. Then from any element of $RS_n$ or $RC_n$, we can construct a configuration of the corresponding puzzle by applying that element to $C_0$. In other words, every transformation $t \in RS_n$ or $t \in RC_n$ corresponds with the configuration $C_t = t(C_0)$ of the $n \times n$ Rubik's Square or $n \times n \times n$ Rubik's Cube that is obtained by applying $t$ to $C_0$. 

Using this idea, we define a new series of problems:

\begin{problem}
The \textbf{Group Rubik's Square} problem has as input a transformation $t \in RS_n$ and a value $k$. The goal is to decide whether the transformation $t$ can be reversed by a sequence of at most $k$ transformations corresponding to Rubik's Square moves. In other words, the answer is ``yes'' if and only if the transformation $t$ can be reversed by a sequence of at most $k$ transformations of the form $x_i$ or $y_i$.
\end{problem}

\begin{problem}
The \textbf{Group STM/SQTM Rubik's Cube} problem has as input a transformation $t \in RC_n$ and a value $k$. The goal is to decide whether the transformation $t$ can be reversed by a sequence of at most $k$ transformations corresponding with legal Rubik's Cube moves under move count metric STM/SQTM. 
\end{problem}

We can interpret these problems as variants of the Rubik's Square or Rubik's Cube problems. For example, the Rubik's Square problem asks whether it is possible (in a given number of moves) to unscramble a Rubik's Square configuration so that each face ends up monochromatic, while the Group Rubik's Square problem asks whether it is possible (in a given number of moves) to unscramble a Rubik's Square configuration so that each sticker goes back to its exact position in the originally solved configuration $C_0$. As you see, the Group Rubik's Square problem, as a puzzle, is just a more difficult variant of the puzzle: instead of asking the player to move all the stickers of the same color to the same face, this variant asks the player to move each stickers to the exact correct position. Similarly, the Group STM/SQTM Rubik's Cube problem as a puzzle asks the player to move each sticker to an exact position. These problems can have practical applications with physical puzzles. For example, some Rubik's Cubes have pictures split up over the stickers of each face instead of just monochromatic colors on the stickers. For these puzzles, as long as no two stickers are the same, the Group STM/SQTM Rubik's Cube problem is more applicable than the STM/SQTM Rubik's Cube problem (which can leave a face ``monochromatic'' but scrambled in image).

We formalize the idea that the Group version of the puzzle is a strictly more difficult puzzle in the following lemmas:

\begin{lemma}
\label{lemma:square_types}
If $(t, k)$ is a ``yes'' instance to the Group Rubik's Square problem, then $(t(C_0), k)$ is a ``yes'' instance to the Rubik's Square problem.
\end{lemma}

\begin{lemma}
\label{lemma:cube_types_1}
If $(t, k)$ is a ``yes'' instance to the Group STM/SQTM Rubik's Cube problem, then $(t(C_0), k)$ is a ``yes'' instance to the STM/SQTM Rubik's Cube problem.
\end{lemma}

The proof of each of these lemmas is the same. If $(t, k)$ is a ``yes'' instance to the Group variants of the puzzle problems, then $t$ can be inverted using at most $k$ elements corresponding to moves. Applying exactly those moves to $t(C_0)$ yields configuration $C_0$, which is a solved configuration of the cube. Thus it is possible to solve the puzzle in configuration $t(C_0)$ in at most $k$ moves. In other words, $(t(C_0), k)$ is a ``yes'' instance to the non-Group variant of the puzzle problem.

At this point it is also worth mentioning that the Rubik's Square with SQTM move model is a strictly more difficult puzzle than the Rubik's Square with STM move model:

\begin{lemma}
\label{lemma:cube_types_2}
If $(C, k)$ is a ``yes'' instance to the SQTM Rubik's Cube problem, then it is also a ``yes'' instance to the STM Rubik's Cube problem. Similarly, if $(t, k)$ is a ``yes'' instance to the Group SQTM Rubik's Cube problem, then it is also a ``yes'' instance to the Group STM Rubik's Cube problem.
\end{lemma}

To prove this lemma, note that every move in the SQTM move model is a legal move in the STM move model. Then if configuration $C$ can be solved in $k$ or fewer SQTM moves, it can certainly also be solved in $k$ or fewer STM moves. Similarly, if $t$ can be inverted using at most $k$ transformations corresponding to SQTM moves, then it can also be inverted using at most $k$ transformations corresponding to STM moves.

\subsection{Membership in NP}

Consider the graph whose vertices are transformations in $RS_n$ (or $RC_n$) and whose edges $(a,b)$ connect transformations $a$ and $b$ for which $a^{-1}b$ is the transformation corresponding to a single move (under the standard Rubik's Square move model or under the STM or SQTM move model). It was shown in \cite{demaine} that the diameter of this graph is $\Theta(\frac{n^2}{\log n})$. This means that any achievable transformation of the puzzle (any transformation in $RS_n$ or $RC_n$) can be reached using a polynomial $p(n)$ number of moves. 

Using this fact, we can build an NP algorithm solving the (Group) STM/SQTM Rubik's Cube and the (Group) Rubik's Square problems. In these problems, we are given $k$ and either a starting configuration or a transformation, and we are asked whether it is possible to solve the configuration/invert the transformation in at most $k$ moves. The NP algorithm can nondeterministically make $\min(k, p(n))$ moves and simply check whether this move sequence inverts the given transformation or solves the given puzzle configuration. 

If any branch accepts, then certainly the answer to the problem is ``yes'' (since that branch's chosen sequence of moves is a solving/inverting sequence of moves of length at most $k$). On the other hand, if there is a solving/inverting sequence of moves of length at most $k$, then there is also one that has length both at most $k$ and at most $p(n)$. This is because $p(n)$ is an upper bound on the diameter of the graph described above. Thus, if the answer to the problem is ``yes'', then there exists a solving/inverting sequence of moves of length at most $\min(k, p(n))$, and so at least one branch accepts. As desired, the algorithm described is correct. Therefore, we have established membership in NP for the problems in question.

\section{Hamiltonicity variants}
\label{section:promise_problems}

To prove the problems introduced above hard, we need to introduce several variants of the Hamiltonian cycle and path problems.

It is shown in \cite{itai} that the following problem is NP-complete. 

\begin{problem}
A \emph{square grid graph} is a finite induced subgraph of the infinite square lattice. The \emph{Grid Graph Hamiltonian Cycle} problem asks whether a given square grid graph with no degree-$1$ vertices has a Hamiltonian cycle.
\end{problem}

Starting with this problem, we prove that the following promise version of the grid graph Hamiltonian path problem is also NP-hard.

\begin{problem}
\label{prob:promise_grid_graph_hamiltonian_path}
The \emph{Promise Grid Graph Hamiltonian Path} problem takes as input a square grid graph $G$ and two specified vertices $s$ and $t$ with the promise that any Hamiltonian path in $G$ has $s$ and $t$ as its start and end respectively. The problem asks whether there exists a Hamiltonian path in $G$.
\end{problem}

The above problem is more useful, but it is still inconvenient in some ways. In particular, there is no conceptually simple way to connect a grid graph to a Rubik's Square or Rubik's Cube puzzle. It is the case, however, that every grid graph is actually a type of graph called a ``cubical graph''. Cubical graphs, unlike grid graphs, can be conceptually related to Rubik's Cubes and Rubik's Squares with little trouble. 

So what is a cubical graph? Let $H_m$ be the $m$ dimensional hypercube graph; in particular, the vertices of $H_m$ are the bitstrings of length $m$ and the edges connect pairs of bitstrings whose Hamming distance is exactly one. Then a \emph{cubical graph} is any induced subgraph of any hypercube graph $H_m$.

Notably, when embedding a grid graph into a hypercube, it is always possible to assign the bitstring label $00\ldots0$ to any vertex. Suppose we start with Promise Grid Graph Hamiltonian Path problem instance $(G, s, t)$; then by embedding $G$ into a hypercube graph, we can reinterpret this instance as an instance of the promise version of cubical Hamiltonian path:

\begin{problem}
\label{prob:promise_cubical_hamiltonian_path}
The \emph{Promise Cubical Hamiltonian Path} problem takes as input a cubical graph whose vertices are length-$m$ bitstrings $l_1, l_2, \ldots, l_n$ with the promise that (1) $l_n = 00\ldots0$ and (2) any Hamiltonian path in the graph has $l_1$ and $l_n$ as its start and end respectively. The problem asks whether there exists a Hamiltonian path in the cubical graph. In other words, the problem asks whether it is possible to rearrange bitstrings $l_1, \ldots, l_n$ into a new order such that each bitstring has Hamming distance one from the next.
\end{problem}

In the remainder of this section, we prove that Problems~\ref{prob:promise_grid_graph_hamiltonian_path} and~\ref{prob:promise_cubical_hamiltonian_path} are NP-hard.

\subsection{Promise Grid Graph Hamiltonian Path is NP-hard}

First, we reduce from the Grid Graph Hamiltonian Cycle problem to the Promise Grid Graph Hamiltonian Path problem. 

\begin{lemma}
The Promise Grid Graph Hamiltonian Path problem (Problem~\ref{prob:promise_grid_graph_hamiltonian_path}) is NP-hard.
\end{lemma}

\begin{proof}
Consider an instance $G$ of the Grid Graph Hamiltonian Cycle problem. Consider the vertices in the top row of $G$ and let the leftmost vertex in this row be $u$. $u$ has no neighbors on its left or above it, so it must have a neighbor to its right (since $G$ has no degree-$1$ vertices). Let that vertex be $u'$. We can add vertices to $G$ above $u$ and $u'$ as shown in figure~\ref{fig:cycle_to_path_addition} to obtain new grid graph $G'$ in polynomial time. Note that two of the added vertices are labeled $v$ and $v'$. Also note that the only edges that are added are those shown in the figure since no vertices in $G$ are above $u$.

\begin{figure}[h]
\centering
\includegraphics[width=.3\textwidth]{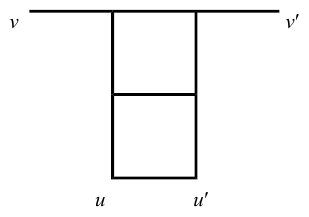}
\caption{The vertices added to $G$ to obtain $G'$.}
\label{fig:cycle_to_path_addition}
\end{figure}

First notice that $(G', v, v')$ is a valid instance of the Promise Grid Graph Hamiltonian Path problem. In particular, $(G', v, v')$ satisfies the promise---any Hamiltonian path in $G'$ must have $v$ and $v'$ as endpoints---since both $v$ and $v'$ have degree-$1$.

Below we show that $(G', v, v')$ is a ``yes'' instance to the Promise Grid Graph Hamiltonian Path problem (i.e., $G'$ has a Hamiltonian path) if and only if $G$ is a ``yes'' instance to the Grid Graph Hamiltonian Cycle problem (i.e., $G$ has a Hamiltonian cycle).

First suppose $G$ contains a Hamiltonian cycle. This cycle necessarily contains edge $(u, u')$ because $u$ has only two neighbors; removing this edge yields a Hamiltonian path from $u'$ to $u$ in $G$. This path can be extended by adding paths from $v'$ to $u'$ and from $u$ to $v$ into a Hamiltonian path in $G'$ from $v'$ to $v$.

On the other hand, suppose $G'$ has a Hamiltonian path. Such a path must have $v$ and $v'$ as the two endpoints, and it is easy to show that the two short paths between $u$ and $v$ and between $u'$ and $v'$ must be the start and end of this path. In other words, if $G'$ has a Hamiltonian path, then the central part of this path is a Hamiltonian path in $G'$ between $u$ and $u'$. Adding edge $(u,u')$, we obtain a Hamiltonian cycle in $G$.

By the above reduction, the Promise Grid Graph Hamiltonian Path problem is NP-hard.
\end{proof}

\subsection{Promise Cubical Hamiltonian Path is NP-hard}
\label{section:promise_cubical_hamiltonian_path}

Second, we reduce from the Promise Grid Graph Hamiltonian Path problem to the Promise Cubical Hamiltonian Path problem. 

\begin{theorem}
The Promise Cubical Hamiltonian Path problem (Problem~\ref{prob:promise_cubical_hamiltonian_path}) is NP-hard.
\end{theorem}

\begin{proof}
Consider an instance $(G, s, t)$ of the Promise Grid Graph Hamiltonian Path problem. Suppose $G$ has $m_r$ rows and $m_c$ columns and $n$ vertices.

Assign a bitstring label to each row and a bitstring label to each column. In particular, let the row labels from left to right be the following length $m_r - 1$ bitstrings: $000\ldots0$, $100\ldots0$, $110\ldots0$, $\ldots$, and $111\ldots1$. Similarly, let the column labels from top to bottom be the following length $m_c - 1$ bitstrings: $000\ldots0$, $100\ldots0$, $110\ldots0$, $\ldots$, and $111\ldots1$. Then assign each vertex a bitstring label of length $m = m_r + m_c - 2$ consisting of the concatenation of its row label followed by its column label.

Consider any two vertices. Their labels have Hamming distance one if and only if the vertices' column labels are the same and their row labels have Hamming distance one, or visa versa. By construction, two row/column labels are the same if and only if the two rows/columns are the same and they have Hamming distance one if and only if the two rows/columns are adjacent. Thus two vertices' labels have Hamming distance one if and only if the two vertices are adjacent in $G$. 

In other words, we have expressed $G$ as a cubical graph by assigning these bitstring labels to the vertices of $G$. In particular, suppose the vertices of $G$ are $v_1, v_2, \ldots, v_n$ with $v_1 = s$ and $v_n = t$. Let $l_i'$ be the label of $v_i$. Then the bitstrings $l_1', l_2', \ldots, l_n'$ specify the cubical graph that is $G$. 

Define $l_i = l_i' \oplus l_n'$. Under this definition, the Hamming distance between $l_i$ and $l_j$ is the same as the Hamming distance between $l_i'$ and $l_j'$. Therefore $l_i$ has Hamming distance one from $l_j$ if and only if $v_i$ and $v_j$ are adjacent. Thus, the cubical graph specified by bitstrings $l_1, \ldots, l_n$ is also $G$. Note that the $l_i$ bitstrings can be computed in polynomial time.

We claim that $l_1, \ldots, l_n$ is a valid instance of Promise Cubical Hamiltonian Path, i.e., this instance satisfies the promise of the problem. The first promise is that $l_n = 00\ldots0$; by definition, $l_n = l_n' \oplus l_n' = 00\ldots0$. The second promise is that any Hamiltonian path in the cubical graph specified by $l_1, l_2, \ldots, l_n$ has $l_1$ and $l_n$ as its start and end. Note that the cubical graph specified by $l_1, l_2, \ldots, l_n$ is the graph $G$ with vertex $l_i$ in the cubical graph corresponding to vertex $v_i$ in $G$. In other words, the promise requested is that any Hamiltonian path in $G$ must start and end in vertices $v_1 = s$ and $v_n = t$. This is guaranteed by the promise of the Promise Grid Graph Hamiltonian Path problem. 

Since $G$ is the graph specified by $l_1, l_2, \ldots, l_n$, the answer to the Promise Cubical Hamiltonian Path instance $l_1, l_2, \ldots, l_n$ is the same as the answer to the Promise Grid Graph Hamiltonian Path instance $(G, s, t)$. Thus, the procedure converting $(G, s, t)$ into $l_1, l_2, \ldots, l_n$, which runs in polynomial time, is a reduction proving that Promise Cubical Hamiltonian Path is NP-hard.
\end{proof}

\section{(Group) Rubik's Square is NP-complete}
\label{section:rubiks_square}

\subsection{Reductions}

To prove that the Rubik's Square and Group Rubik's Square problems are NP-complete, we reduce from the Promise Cubical Hamiltonian Path problem of Section~\ref{section:promise_cubical_hamiltonian_path}.

Suppose we are given an instance of the Promise Cubical Hamiltonian Path problem consisting of $n$ bitstrings $l_1, \ldots, l_n$ of length $m$ (with $l_n = 00\ldots0$). To construct a Group Rubik's Square instance we need to compute the value $k$ indicating the allowed number of moves and construct the transformation $t \in RS_s$.

The value $k$ can be computed directly as $k = 2n - 1$.

The transformation $t$ will be an element of group $RS_s$ where $s = 2(\max(m,n) + 2n)$. Define $a_i$ for $1\le i \le n$ to be $(x_1)^{(l_i)_1}\circ(x_2)^{(l_i)_2}\circ\cdots\circ(x_m)^{(l_i)_m}$ where $(l_i)_1, (l_i)_2, \ldots, (l_i)_m$ are the bits of $l_i$. Also define $b_i = (a_i)^{-1} \circ y_i \circ a_i$ for $1\le i \le n$. Then we define $t$ to be $a_1\circ b_1 \circ b_2 \circ\cdots\circ b_n$.

Outputting $(t, k)$ completes the reduction from the Promise Cubical Hamiltonian Path problem to the Group Rubik's Square problem. To reduce from the Promise Cubical Hamiltonian Path problem to the Rubik's Square problem we simply output $(C_t, k) = (t(C_0), k)$. These reductions clearly run in polynomial time.

\subsection{Intuition}

The key idea that makes this reduction work is that the transformations $b_i$ for $i \in \{1, \ldots, n\}$ all commute. This allows us to rewrite $t = a_1\circ b_1 \circ b_2 \circ\cdots\circ b_n$ with the $b_i$s in a different order. If the order we choose happens to correspond to a Hamiltonian path in the cubical graph specified by $l_1, \ldots, l_n$, then when we explicitly write the $b_i$s and $a_1$ in terms of $x_j$s and $y_i$s, most of the terms cancel. In particular, the number of remaining terms will be exactly $k$. Since we can write $t$ as a combination of exactly $k$ $x_j$s and $y_i$s, we can invert $t$ using at most $k$ $x_j$s and $y_i$s. In other words, if there is a Hamiltonian path in the cubical graph specified by $l_1, \ldots, l_n$, then $(t, k)$ is a ``yes'' instance to the Group Rubik's Square problem.

In order to more precisely describe the cancellation of terms in $t$, we can consider just one local part: $b_{i} \circ b_{i'}$. We can rewrite this as $(a_{i})^{-1} \circ y_{i} \circ a_{i} \circ (a_{i'})^{-1} \circ y_{i'} \circ a_{i'}$. The interesting part is that $a_{i} \circ (a_{i'})^{-1}$ will cancel to become just one $x_j$. Note that 
$$a_{i} \circ (a_{i'})^{-1} = (x_1)^{(l_{i})_1}\circ(x_2)^{(l_{i})_2}\circ\cdots\circ(x_m)^{(l_{i})_m} \circ (x_1)^{-(l_{i'})_1}\circ(x_2)^{-(l_{i'})_2}\circ\cdots\circ(x_m)^{-(l_{i'})_m},$$
which we can rearrange as 
$$(x_1)^{(l_{i})_1 - (l_{i'})_1}\circ(x_2)^{(l_{i})_2 - (l_{i'})_2}\circ\cdots\circ(x_m)^{(l_{i})_m - (l_{i'})_m}.$$
Next, if $b_{i}$ and $b_{i'}$ correspond to adjacent vertices $l_{i}$ and $l_{i'}$, then $(l_{i})_j - (l_{i'})_j$ is zero for all $j$ except one for which $(l_{i})_j - (l_{i'})_j = \pm 1$. Thus the above can be rewritten as $(x_j)^{1}$ or $(x_j)^{-1}$ for some specific $j$. Since $x_j = (x_j)^{-1}$ this shows that $(a_{i_1})^{-1} \circ a_{i_2}$ simplifies to $x_j$ for some $j$.

This intuition is formalized in a proof in the following subsection.

\subsection{Promise Cubical Hamiltonian Path solution $\to$ (Group) Rubik's Square solution}
\label{section:square_first_direction}

\begin{lemma}
The transformations $b_i$ all commute.
\label{lemma:square_b_i_commute}
\end{lemma}

\begin{proof}
Consider any such transformation $b_i$. The transformation $b_i$ can be rewritten as $(a_i)^{-1} \circ y_i \circ a_i$. For any cubie not moved by the $y_i$ middle term, the effect of this transformation is the same as the effect of transformation $(a_i)^{-1} \circ a_i = 1$. In other words, $b_i$ only affects cubies that are moved by the $y_i$ term. But $y_i$ only affects cubies with $y$ coordinate $i$. In general in a Rubik's Square, cubies with $y$ coordinate $i$ at some particular time will have $y$ coordinate $\pm i$ at all times. Thus, all the cubies affected by $b_i$ start in rows $\pm i$. 

This is enough to see that the cubies affected by $b_i$ are disjoint from those affected by $b_j$ (for $j \ne i$). In other words, the transformations $b_i$ all commute.
\end{proof}

\begin{theorem}
If $l_1, \ldots, l_n$ is a ``yes'' instance to the Promise Cubical Hamiltonian Path problem, then $(t, k)$ is a ``yes'' instance to the Group Rubik's Square problem.
\label{thm:square_first_direction}
\end{theorem}

\begin{proof}
Suppose $l_1, \ldots, l_n$ is a ``yes'' instance to the Promise Cubical Hamiltonian Path problem. Let $m$ be the length of $l_i$ and note that $l_n = 00\ldots0$ by the promise of the Promise Cubical Hamiltonian Path problem. Furthermore, since $l_1, \ldots, l_n$ is a ``yes'' instance to the Promise Cubical Hamiltonian Path problem, there exists an ordering of these bitstrings $l_{i_1}, l_{i_2}, \ldots, l_{i_n}$ such that each consecutive pair of bitstrings is at Hamming distance one, $i_1 = 1$, and $i_n = n$ (with the final two conditions coming from the promise).

By Lemma~\ref{lemma:square_b_i_commute}, we know that $t = a_1 \circ b_1 \circ b_2 \circ\cdots \circ b_n$ can be rewritten as 
$$t = a_1 \circ b_{i_1} \circ b_{i_2} \circ\cdots \circ b_{i_n}.$$ 
Using the definition of $b_i$, we can further rewrite this as 
$$t = a_1 \circ ((a_{i_1})^{-1} \circ y_{i_1} \circ a_{i_1}) \circ ((a_{i_2})^{-1} \circ y_{i_2} \circ a_{i_2}) \circ\cdots \circ ((a_{i_n})^{-1} \circ y_{i_n} \circ a_{i_n}),$$
or as 
$$t = (a_1 \circ (a_{i_1})^{-1}) \circ y_{i_1} \circ (a_{i_1} \circ (a_{i_2})^{-1}) \circ y_{i_2} \circ (a_{i_2} \circ (a_{i_3})^{-1}) \circ\cdots \circ (a_{i_{n-1}} \circ (a_{i_n})^{-1}) \circ y_{i_n} \circ (a_{i_n}).$$

We know that $i_1 = 1$, and therefore that $a_1 \circ (a_{i_1})^{-1} = a_1 \circ (a_1)^{-1} = 1$ is the identity element. Similarly, we know that $i_n = n$ and therefore that $a_{i_n} = a_n = (x_1)^{(l_n)_1}\circ(x_2)^{(l_n)_2}\circ\cdots\circ(x_m)^{(l_n)_m} =  (x_1)^{0}\circ(x_2)^{0}\circ\cdots\circ(x_m)^{0} = 1$ is also the identity.

Thus we see that $$t = y_{i_1} \circ (a_{i_1} \circ (a_{i_2})^{-1}) \circ y_{i_2} \circ (a_{i_2} \circ (a_{i_3})^{-1}) \circ\cdots \circ (a_{i_{n-1}} \circ (a_{i_n})^{-1}) \circ y_{i_n}.$$

Consider the transformation $a_{i_p} \circ (a_{i_{p+1}})^{-1}$. This transformation can be written as $$a_{i_p} \circ (a_{i_{p+1}})^{-1} = (x_1)^{(l_{i_p})_1} \circ (x_2)^{(l_{i_p})_2} \circ\cdots \circ (x_m)^{(l_{i_p})_m} \circ (x_1)^{-(l_{i_{p+1}})_1} \circ (x_2)^{-(l_{i_{p+1}})_2} \circ\cdots \circ (x_m)^{-(l_{i_{p+1}})_m}.$$

Because $x_u$ always commutes with $x_v$, we can rewrite this as $$a_{i_p} \circ (a_{i_{p+1}})^{-1} = (x_1)^{(l_{i_p})_1 - (l_{i_{p+1}})_1} \circ (x_2)^{(l_{i_p})_2-(l_{i_{p+1}})_2} \circ\cdots \circ (x_m)^{(l_{i_p})_m-(l_{i_{p+1}})_m}.$$

Since $l_{i_p}$ differs from $l_{i_{p+1}}$ in only one position, call it $j_p$, we see that $(l_{i_p})_j-(l_{i_{p+1}})_j$ is zero unless $j = j_p$, and is $\pm 1$ in that final case. This is sufficient to show that $a_{i_p} \circ (a_{i_{p+1}})^{-1} = (x_{j_p})^{\pm 1} = x_{j_p}$.

Thus we see that $$t = y_{i_1} \circ x_{j_1} \circ y_{i_2} \circ x_{j_2} \circ\cdots \circ x_{j_{n-1}} \circ y_{i_n},$$
or (by left multiplying) that 
$$1 = y_{i_n}^{-1} \circ x_{j_{n-1}}^{-1} \circ\cdots \circ x_{j_2}^{-1} \circ y_{i_2}^{-1} \circ x_{j_1}^{-1} \circ y_{i_1}^{-1} \circ t = y_{i_n} \circ x_{j_{n-1}} \circ\cdots \circ x_{j_2} \circ y_{i_2} \circ x_{j_1} \circ y_{i_1} \circ t.$$ 
We see that $t$ can be reversed by $k = 2n-1$ moves of the form $x_j$ or $y_i$, or in other words that $(t, k)$ is a ``yes'' instance to the Group Rubik's Square problem.
\end{proof}

\begin{corollary}
If $l_1, \ldots, l_n$ is a ``yes'' instance to the Promise Cubical Hamiltonian Path problem, then $(C_t, k)$ is a ``yes'' instance to the Rubik's Square problem.
\label{corollary:square_first_direction}
\end{corollary}

\begin{proof}
This follows immediately from Theorem~\ref{thm:square_first_direction} and Lemma~\ref{lemma:square_types}.
\end{proof}

\subsection{\boldmath Coloring of $C_t$}

In order to show the other direction of the proof, it will be helpful to consider the coloring of the stickers on the top and bottom faces of the Rubik's Square. In particular, if we define $b = b_1 \circ\cdots \circ b_n$ (so that $t = a_1 \circ b$), then it will be very helpful for us to know the colors of the top and bottom stickers in configuration $C_b = b(C_0)$.

Consider for example the instance of Promise Cubical Hamiltonian Path with $n = 5$ and $m = 3$ defined below:
\begin{align*}
l_1 &= 011 \\
l_2 &= 110 \\
l_3 &= 111 \\
l_4 &= 100 \\
l_5 &= 000
\end{align*}
For this example, $C_0$ is an $s \times s$ Rubik's Square with $s = 2(\max(m,n)+2n) = 30$.

To describe configuration $C_b$, we need to know the effect of transformation $b_i$. For example, Figure~\ref{fig:square_example_1} shows the top face of a Rubik's Square in configurations $C_0$, $a_2(C_0)$, $(y_2 \circ a_2)(C_0)$, and $b_2(C_0) = ((a_2)^{-1} \circ y_2 \circ a_2)(C_0)$ where $a_2$ and $y_2$ are defined in terms of $l_2 = 110$ as in the reduction.

\begin{figure}[h]
\centering
\begin{subfigure}[b]{0.48\textwidth}
\includegraphics[width=\textwidth]{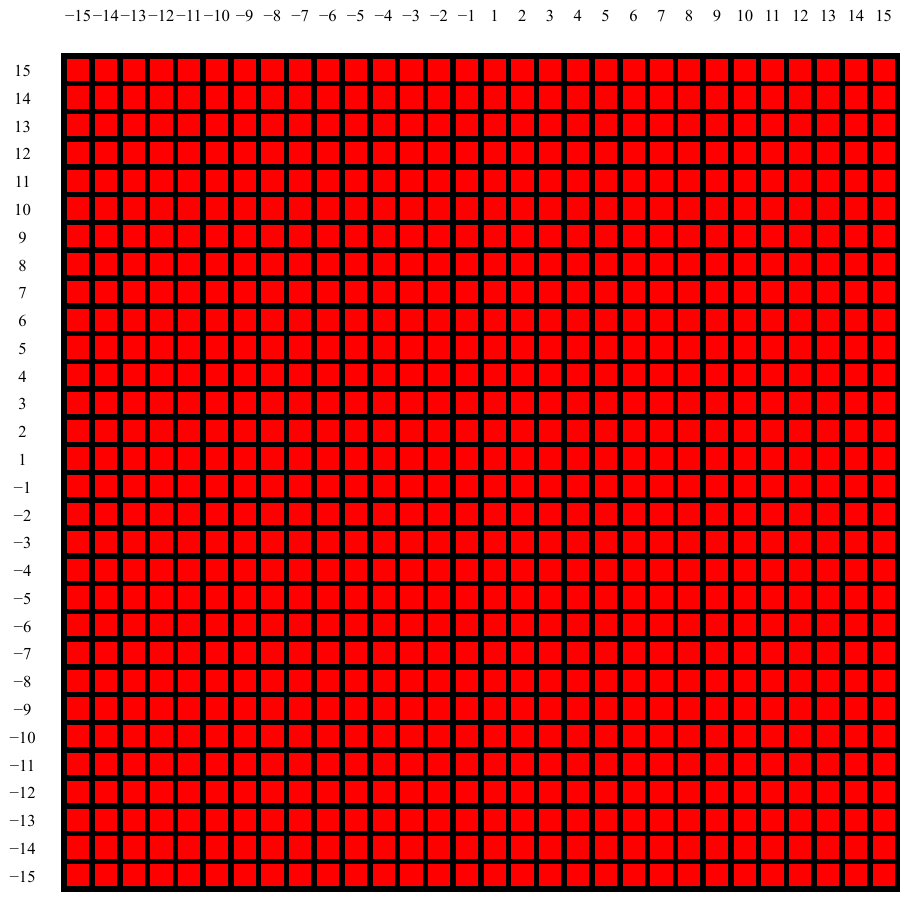}
\caption{\label{fig:square_example_1a}}
\end{subfigure}
\hfill
\begin{subfigure}[b]{0.48\textwidth}
\includegraphics[width=\textwidth]{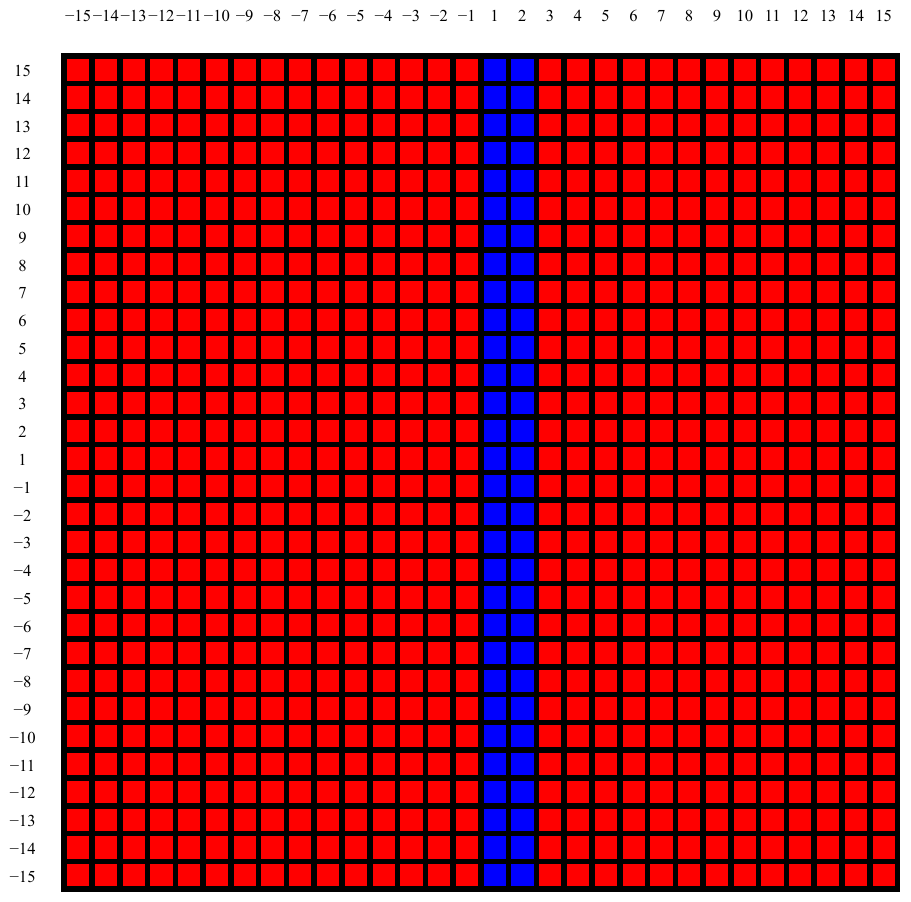}
\caption{\label{fig:square_example_1b}}
\end{subfigure}

\begin{subfigure}[b]{0.48\textwidth}
\includegraphics[width=\textwidth]{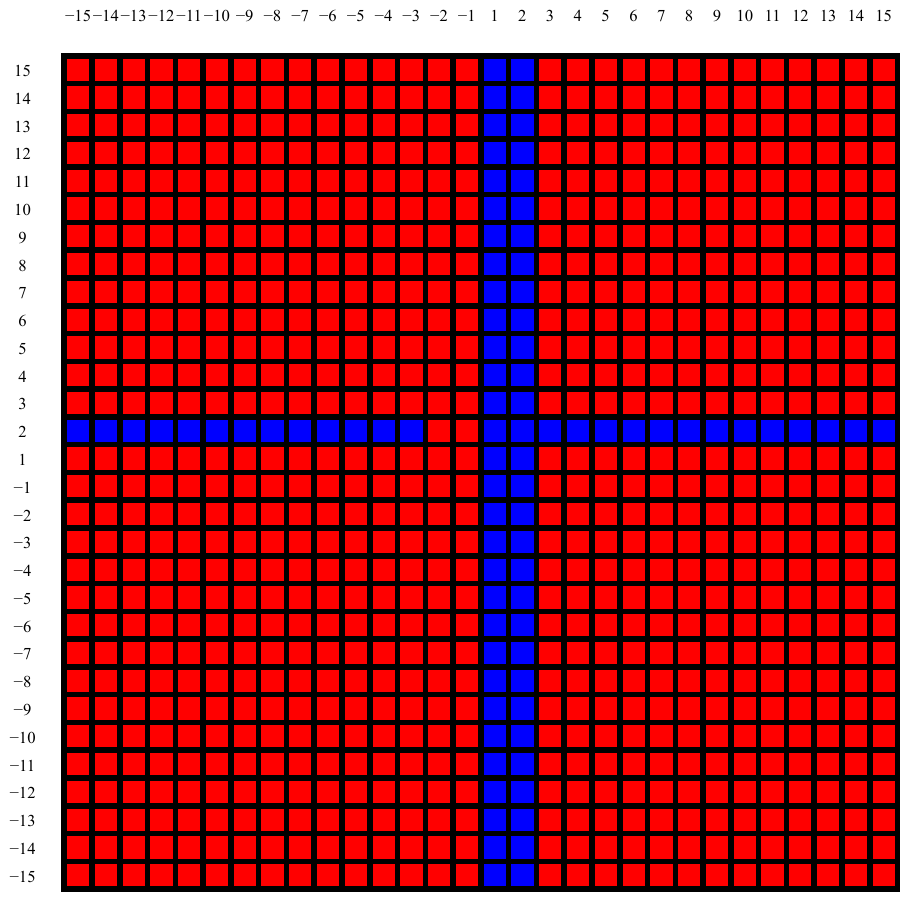}
\caption{\label{fig:square_example_1c}}
\end{subfigure}
\hfill
\begin{subfigure}[b]{0.48\textwidth}
\includegraphics[width=\textwidth]{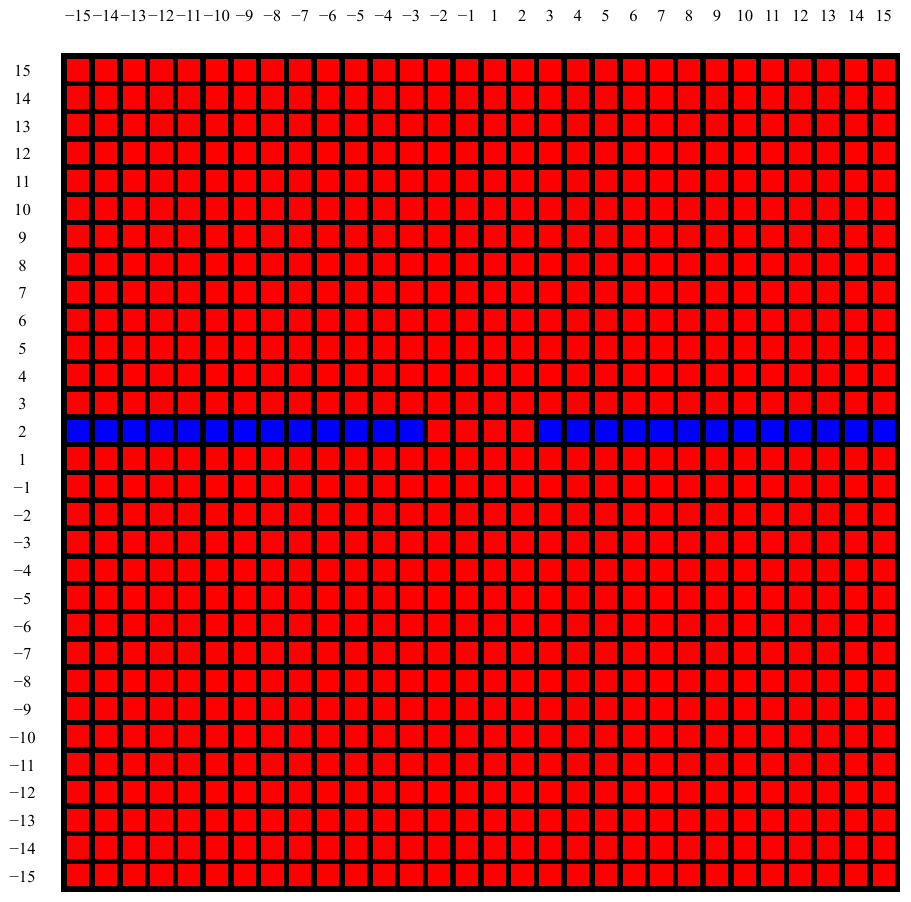}
\caption{\label{fig:square_example_1d}}
\end{subfigure}
\caption{Applying $b_2$ to $C_0$ step by step (only top face shown).}
\label{fig:square_example_1}
\end{figure}

The exact behavior of a Rubik's Square due to $b_i$ is described by the following lemma:

\begin{lemma}
Suppose $i \in \{1,\ldots,n\}$, and $c, r \in \{1, \ldots, s/2\}$. Then 
\begin{enumerate}
\item if $r = i$ and $c \le m$ such that bit $c$ of $l_i$ is $1$, then $b_i$ swaps the cubies in positions $(c, -r)$ and $(-c, r)$ without flipping either;
\item if $r = i$ and either $c > m$ or $c \le m$ and bit $c$ of $l_i$ is $0$, then $b_i$ swaps the cubies in positions $(c, r)$ and $(-c, r)$ and flips them both;
\item all other cubies are not moved by $b_i$.
\end{enumerate}
\label{lemma:square_b_i_effect}
\end{lemma}

\begin{proof}
As noted in the proof of Lemma~\ref{lemma:square_b_i_commute}, a cubie is affected by $b_i = (a_i)^{-1} \circ y_i \circ a_i$ if and only if it is moved by the $y_i$ term.

Note also that $(a_i)^{-1} = a_i$ only moves cubies within their columns and only for columns $c$ for which bit $c$ of $l_i$ is $1$. One consequence is that a cubie can only be moved by $a_i$ if its column index is positive. Any cubie moved by the $y_i$ term will have a column index of different signs before and after the $y_i$ move, so as a consequence such a cubie cannot be moved by both $a_i$ and $(a_i)^{-1}$. 

Thus there are three possibilities for cubies that are moved by $b_i$: (1) the cubie is moved only by $y_i$, (2) the cubie is moved by $a_i$ and then by $y_i$, and (3) the cubie is moved by $y_i$ and then by $(a_i)^{-1}$.

Consider any cubie of type (1) whose coordinates have absolute values $c$ and $r$. Since the cubie is moved by $y_i$, we know that $r = i$. Since it is not moved by either $a_i$ or $(a_i)^{-1}$, we know that the cubie's column index both before and after the move is not one of the column indices affected by $a_i$. But these two column indices are $c$ and $-c$ (in some order). Therefore it must not be the case that bit $c$ of $l_i$ is $1$. Also note that cubies of this type are flipped exactly once. Putting that together, we see that if $c \in \{1, \ldots, s/2\}$, $r = i$, and it is not the case that bit $c$ of $l_i$ exists and is $1$, then $b_i$ swaps the cubies in positions $(c, r)$ and $(-c, r)$ and flips them both.

Consider any cubie of type (2) whose coordinates have absolute values $c$ and $r$. Since the cubie is first moved by $a_i$ and then by $y_i$, we know that $r = i$ and that $c \le m$ with bit $c$ of $l_i$ equal to $1$. Furthermore, the cubie must have started in position $(c, -r)$, then moved to position $(c, r)$ by $a_i$, and then moved to position $(-c, r)$ by $y_i$. Since this cubie is flipped twice, it is overall not flipped. 

Consider on the other hand any cubie of type (3) whose coordinates have absolute values $c$ and $r$. Since the cubie is first moved by $y_i$ and then by $(a_i)^{-1} = a_i$, we know that $r = i$ and that $c \le m$ with bit $c$ of $l_i$ equal to $1$. Furthermore, the cubie must have started in position $(-c, r)$, then moved to position $(c, r)$ by $y_i$, and then moved to position $(c, -r)$ by $a_i$. Since this cubie is flipped twice, it is overall not flipped. 

Putting that together, we see that if $r = i$, and bit $c$ of $l_i$ is $1$, then $b_i$ swaps the cubies in positions $(c, -r)$ and $(-c, r)$ without flipping either.

This covers the three types of cubies that are moved by $b_i$. All other cubies remain in place.
\end{proof}

We can apply the above to figure out the effect of transformation $b_1 \circ b_2 \circ\cdots \circ b_n$ on configuration $C_0$. In particular, that allows us to learn the coloring of configuration $C_b$.

\begin{theorem}
\label{thm:square_coloring}
In $C_b$, a cubie has top face blue if and only if it is in position $(c, r)$ such that $1 \le r \le n$ and either $|c| > m$ or $|c| \le m$ and bit $|c|$ of $l_r$ is $0$.
\end{theorem}

\begin{proof}
$C_b$ is obtained from $C_0$ by applying transformation $b_1 \circ b_2 \circ\cdots \circ b_n$. A cubie has top face blue in $C_b$ if and only if transformation $b_1 \circ b_2 \circ\cdots \circ b_n$ flips that cubie an odd number of times. Each $b_i$ affects a disjoint set of cubies. Thus, among the cubies affected by some particular $b_i$, the only ones that end up blue face up are the ones that are flipped by $b_i$. By Lemma~\ref{lemma:square_b_i_effect}, these are the cubies in row $i$ with column $c$ such that it is not the case that bit $|c|$ of $l_i$ is $1$.  Tallying up those cubies over all the $b_i$s yields exactly the set of blue-face-up cubies given in the theorem statement.
\end{proof}

This concludes the description of $C_b$ in terms of colors. The coloring of configuration $C_t$---the configuration that is actually obtained by applying the reduction to $l_1, \ldots, l_n$---can be obtained from the coloring of configuration $C_b$ by applying transformation $a_1$. 

Applying Theorem~\ref{thm:square_coloring} to the previously given example, we obtain the coloring of the Rubik's Square in configuration $C_b$ as shown in Figure~\ref{fig:square_example_2a}. Note that the $n\times m$ grid of bits comprising $l_1, \ldots, l_n$ is actually directly encoded in the coloring of a section of the Rubik's Square. In addition, the coloring of the Rubik's Square in configuration $C_t$ is shown for the same example in Figure~\ref{fig:square_example_2b}.

\begin{figure}[h]
\centering
\begin{subfigure}[b]{0.48\textwidth}
\includegraphics[width=\textwidth]{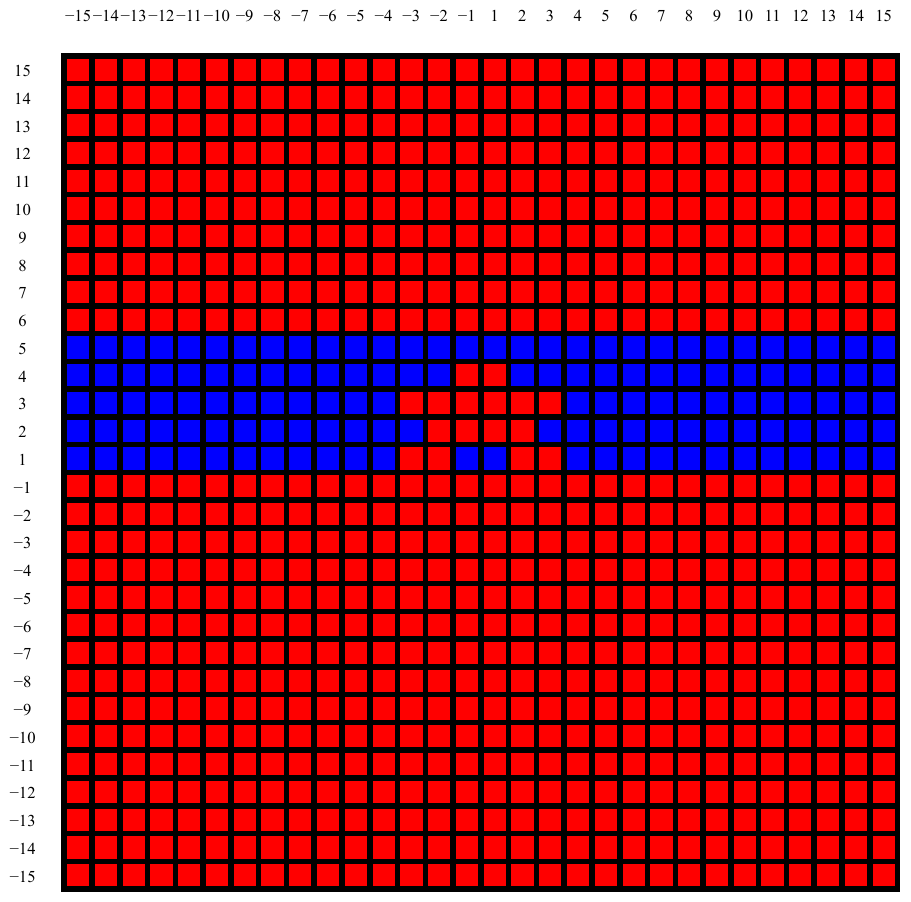}
\caption{The top face of $C_b$ for the example input $l_1, \ldots, l_n$.}
\label{fig:square_example_2a}
\end{subfigure}
\hfill
\begin{subfigure}[b]{0.48\textwidth}
\includegraphics[width=\textwidth]{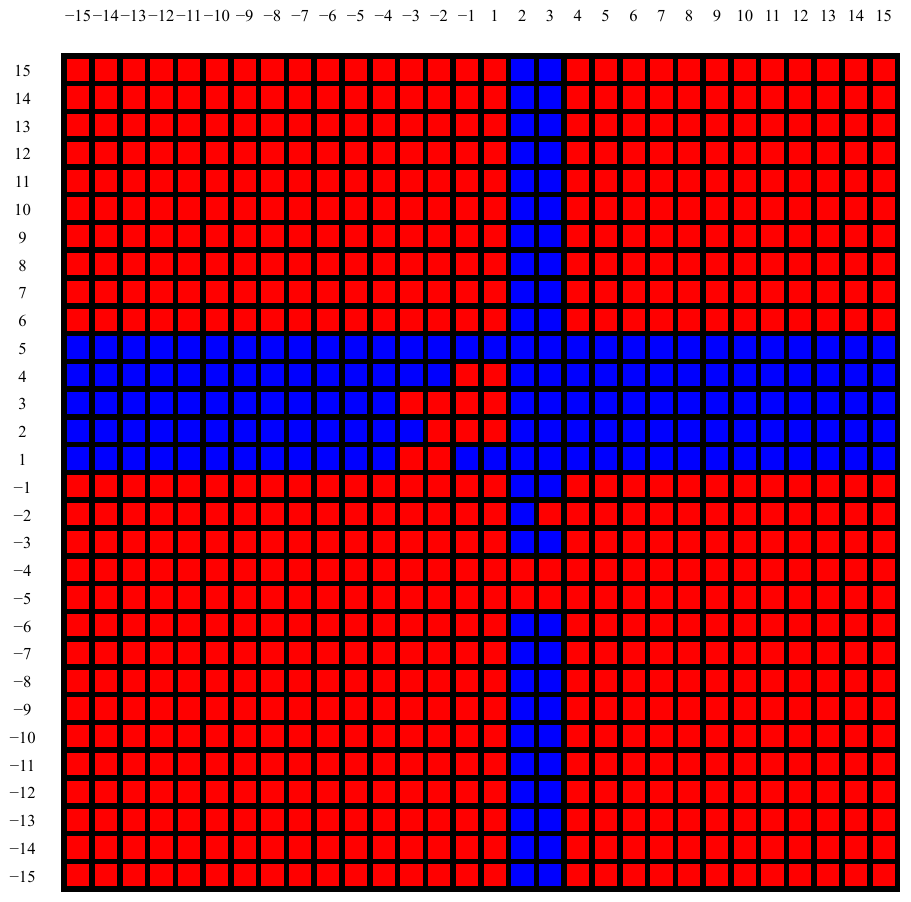}
\caption{The top face of $C_t$ for the example input $l_1, \ldots, l_n$.}
\label{fig:square_example_2b}
\end{subfigure}
\caption{The coloring of the Rubik's Square for the example input $l_1, \ldots, l_n$.}
\end{figure}

\subsection{(Group) Rubik's Square solution $\to$ Promise Cubical Hamiltonian Path solution}

Below, we prove the following theorem:

\begin{theorem}
If $(C_t, k)$ is a ``yes'' instance to the Rubik's Square problem, then $l_1, \ldots, l_n$ is a ``yes'' instance to the Promise Cubical Hamiltonian Path problem.
\label{thm:square_second_direction}
\end{theorem}

By Lemma~\ref{lemma:square_types}, this will immediately also imply the following corollary:

\begin{corollary}
If $(t, k)$ is a ``yes'' instance to the Group Rubik's Square problem, then $l_1, \ldots, l_n$ is a ``yes'' instance to the Promise Cubical Hamiltonian Path problem.
\label{corollary:square_second_direction}
\end{corollary}

To prove the theorem, we consider a hypothetical solution to the $(C_t, k)$ instance of the Rubik's Square problem. A solution consists of a sequence of Rubik's Square moves $m_1, \ldots, m_{k'}$ with $k' \le k$ such that $C' = (m_{k'} \circ\cdots \circ m_1)(C_t)$ is a solved configuration of the Rubik's Square. Throughout the proof, we will use only the fact that move sequence $m_1, \ldots, m_{k'}$ solves the top and bottom faces of the Rubik's Square in configuration $C_t$.

The main idea of the proof relies on three major steps. In the first step, we show that $m_1, \ldots, m_{k'}$ must flip row $i$ an odd number of times if $i \in \{1, \ldots, n\}$, and an even number of times otherwise.

We then define set $O\subseteq\{1, \ldots,n\}$ (where $O$ stands for ``one'') to be the set of indices $i$ such that there is exactly one index-$i$ row move. Clearly, in order to satisfy the parity constraints, every $i \in O$ must have one row $i$ move and zero row $-i$ moves in $m_1, \ldots, m_{k'}$. The second step of the proof is to show that, if $i_1, i_2 \in O$, then the number of column moves in $m_1, \ldots, m_{k'}$ between the single flip of row $i_1$ and the single flip of row $i_2$ is at least the Hamming distance between $l_{i_1}$ and $l_{i_2}$.

The final step of the proof is a counting argument. There are four types of moves in $m_1, \ldots, m_{k'}$: 
\begin{enumerate}
\item index-$i$ row moves with $i \in O$ (all of which are positive moves as shown above),
\item index-$i$ row moves with $i \in \{1, \ldots, n\} \setminus O$,
\item column moves, and
\item index-$i$ row moves with $i \not\in \{1, \ldots, n\}$.
\end{enumerate}

For each $i \in O$, there is exactly one index-$i$ move by definition of $O$. Therefore the number of type-$1$ moves is exactly $|O|$.

For each $i$ in $\{1, \ldots, n\} \setminus O$, the number of index-$i$ row moves is odd by the parity constraint. Furthermore, by the definition of $O$, this number is not one. Thus each $i$ in $\{1, \ldots, n\} \setminus O$ contributes at least three moves. Therefore the number of type-$2$ moves is at least $3(|\{1, \ldots, n\} \setminus O|) = 3(n - |O|)$. 

Consider the moves of rows $i$ with $i \in O$. Since the $l_i$s are all distinct, there must be at least one column move between every consecutive pair of such moves. Thus the total number of type-$3$ moves (column moves) is at least $|O| - 1$. Furthermore, the number of type-$3$ moves is $|O| - 1$ if and only if the consecutive pairs of row $i \in O$ moves have exactly one column move between them. Such a pair of $i$s has exactly one column move between the two row-$i$ moves only if the corresponding pair of $l_i$s is at Hamming distance one. Therefore, if we consider the $l_i$s for $i \in O$ in the order in which row-$i$ moves occur in $m_1, \ldots, m_{k'}$, then the number of type-$3$ moves is exactly $|O| - 1$ if and only if those $l_i$s in that order have each $l_i$ at Hamming distance exactly one from the next (and more otherwise).

The number of type-$4$ moves is at least $0$.

Adding these bounds up, we see that there are at least $(|O|) + 3(n - |O|) + (|O| - 1) + 0 = 3n - 1 - |O| = k + (n - |O|)$ moves. Since $n - |O| \ge 0$ and the number of moves is at most $k$, we can conclude that (1) $|O| = n$ and (2) the number of moves of each type is exactly the minimum possible computed above. Since $|O| = n$ we know that $O = \{1, \ldots, n\}$. But then looking at the condition for obtaining the minimum possible number of type-$3$ moves, we see that the $l_i$s for $i \in O = \{1, \ldots, n\}$ in the order in which row-$i$ flips occur in $m_1, \ldots, m_{k'}$ are each at Hamming distance exactly one from the next. Thus, there is a reordering of $l_1, \ldots, l_n$ in which each $l_i$ is Hamming distance one from the next; in other words, the cubical graph specified by bitstrings $l_1, \ldots, l_n$ has a Hamiltonian path and $l_1, \ldots, l_n$ is a ``yes'' instance to the Promise Cubical Hamiltonian Path problem.

All that's left is to complete the first two steps of the proof. We prove these two steps in the lemmas below:

\begin{lemma}
Move sequence $m_1, \ldots, m_{k'}$ must flip row $i$ an odd number of times if $i \in \{1, \ldots, n\}$, and an even number of times otherwise.
\end{lemma}

\begin{proof}
Consider the transformation 
$$m_{k'}\circ\cdots \circ m_1 \circ t = m_{k'} \circ\cdots \circ m_1 \circ a_1 \circ b_1 \circ b_2 \circ\cdots \circ b_n.$$
This transformation, while not necessarily the identity transformation, must transform $C_0$ into another solved Rubik's Square configuration $C'$. 

Consider the $2n = k+1$ indices $\max(m, n) + 1, \ldots, \max(m,n) + 2n$. At least one such index $i$ must exist for which no move in $m_1, \ldots, m_{k'}$ is an index-$i$ move. Let $u$ be such an index.

Consider the effect of transformation $m_{k'} \circ\cdots \circ m_1 \circ a_1 \circ b_1 \circ b_2 \circ\cdots \circ b_n$ on the cubie in position $(u, u)$. If we write $t = a_1 \circ b_1 \circ b_2 \circ\cdots \circ b_n$ as a sequence of $x_j$s and $y_{i'}$s (using the definitions of $a_1$ and $b_i$), then every move in $t$ flips rows and columns with indices of absolute value at most $\max(m, n)$. Thus no term in the transformation ($m_{k'} \circ\cdots \circ m_1 \circ a_1 \circ b_1 \circ b_2 \circ\cdots \circ b_n$) flips row or column $u$. We conclude that the cubie in position $(u,u)$ is unmoved by this transformation. Applying this transformation to $C_0$ yields $C'$. So since this cubie starts with top sticker red in configuration $C_0$, the final configuration $C'$ also has this cubie's top sticker red. Since $C'$ is a solved configuration, the entire top face in $C'$ must be red.

Next consider the cubie in position $(u, r)$ for any $r$. Since no row or column with index $\pm u$ is ever flipped in transformation $m_{k'} \circ\cdots \circ m_1 \circ a_1 \circ b_1 \circ b_2 \circ\cdots \circ b_n$, this cubie is only ever affected by flips of row $r$. Furthermore, every flip of row $r$ flips this cubie and therefore switches the color of its top face. Since the transformation in question converts configuration $C_0$ into configuration $C'$, both of which have every cubie's top face red, the row in question must be flipped an even number of times.

For $i \in \{1, \ldots, n\}$, the transformation $a_1 \circ b_1 \circ b_2 \circ\cdots \circ b_n$, when written out fully in terms of $y_{i'}$s and $x_j$s, includes exactly one flip of row $y_i$. Thus move sequence $m_1, \ldots, m_{k'}$ must flip each of these rows an odd number of times. Similarly, for $i \not\in \{1, \ldots, n\}$, the transformation $a_1 \circ b_1 \circ b_2 \circ\cdots \circ b_n$, when written out fully in terms of $y_{i'}$s and $x_j$s, does not include any flips of row $y_i$ at all. Thus move sequence $m_1, \ldots, m_{k'}$ must flip each of these rows an even number of times. 
\end{proof}

\begin{lemma}
If $i_1, i_2 \in O$ (with $i_1\ne i_2$), then the number of column moves $x_j$ between the unique $y_{i_1}$ and $y_{i_2}$ moves in sequence $m_1, \ldots, m_{k'}$ is at least the Hamming distance between $l_{i_1}$ and $l_{i_2}$.
\end{lemma}

\begin{proof}
We will prove the following useful fact below: if $i_1, i_2 \in O$ (with $i_1\ne i_2$) and $j \in \{1, 2, \ldots, m\}$ such that the top colors of the cubies in locations $(j, i_1)$ and $(j, i_2)$ are different in configuration $C_b$, then there must be at least one index-$j$ column move in between the unique $y_{i_1}$ and $y_{i_2}$ moves in sequence $m_1, \ldots, m_{k'}$.

We know from Theorem~\ref{thm:square_coloring} that, if $i \in \{1, 2, \ldots, n\}$ and $j \in \{1, 2, \ldots, m\}$, then the top color of the cubie in location $(j, i)$ of configuration $C_b$ is red if and only if $(l_i)_j = 1$. Thus, if $l_{i_1}$ and $l_{i_2}$ differ in bit $j$, then in configuration $C_b$ one of the two cubies in positions $(j, i_1)$ and $(j, i_2)$ will have top face red and the other will have top face blue. Applying the above useful fact, we see that at least one index-$j$ column move will occur in sequence $m_1, \ldots, m_{k'}$ between the unique $y_{i_1}$ and $y_{i_2}$ moves. Since this column move has index $\pm j$, every difference in $l_{i_1}$ and $l_{i_2}$ will contribute at least one distinct column move between the unique $y_{i_1}$ and $y_{i_2}$ moves. Assuming the useful fact, we can conclude that the number of column moves between the unique $y_{i_1}$ and $y_{i_2}$ moves is at least the Hamming distance between $l_{i_1}$ and $l_{i_2}$, as desired.

We now prove the useful fact by contradiction. Assume that the useful fact is false, i.e., that there exists some $i_1, i_2 \in O$ and $j \in \{1, 2, \ldots, m\}$ such that the top colors of the cubies in locations $(j, i_1)$ and $(j, i_2)$ are different in $C_b$ and such that no index-$j$ column move is made between the unique $y_{i_1}$ and $y_{i_2}$ moves in sequence $m_1, \ldots, m_{k'}$.

Consider these two cubies. Starting in configuration $C_b$, we can reach configuration $C'$ by applying transformation $m_{k'} \circ\cdots \circ m_1 \circ a_1 = m_{k'} \circ\cdots \circ m_1 \circ (x_1)^{(l_1)_1} \circ (x_2)^{(l_1)_2} \circ\cdots \circ (x_3)^{(l_1)_m}$. Note that this transformation consists of some (but not necessarily all) of the moves $x_1, x_2, \ldots, x_m$ followed by the move sequence $m_1, \ldots, m_{k'}$. We will consider the effect of this transformation on the two cubies. 

Since the two cubies start in locations $(j, i_1)$ and $(j, i_2)$, the only moves that could ever affect these cubies are of the forms $x_j$, $x_{-j}$, $y_{i_1}$, $y_{-i_1}$, $y_{i_2}$, and $y_{-i_2}$. Furthermore, by the definition of $O$, no moves of the form $y_{-i_1}$ or $y_{-i_2}$ occur and the moves $y_{i_1}$ and $y_{i_2}$ each occur exactly once. Finally, we have by assumption that no moves of the form $x_j$ or $x_{-j}$ (index-$j$ column moves) occur between moves $y_{i_1}$ and $y_{i_2}$. 

Putting these facts together, we see that the effect of transformation $m_{k'} \circ\cdots \circ m_1 \circ a_1$ on these two cubies is exactly the same as the effect of some transformation of the following type: (1) some number of moves of the form $x_j$ or $x_{-j}$, followed by (2) the two moves $y_{i_1}$ and $y_{i_2}$ in some order, followed by (3) some number of moves of the form $x_j$ or $x_{-j}$.

Consider the effect of any such transformation on the two cubies. In step (1), each move of the form $x_j$ or $x_{-j}$ either flips both cubies (since they both start in column $j$) or flips neither, so the two cubies are each flipped an equal number of times. Furthermore, the row index of the two cubies is either positive for both or negative for both at all times throughout step (1). In step (2), either each of the two cubies is flipped exactly once (if their row indices at the start of step (2) are both positive) or neither of the two cubies is flipped at all (if their row indices at the start of step (2) are negative); again, the number of flips is the same. Finally, in step (3), both cubies are in the same column (column $j$ if they were not flipped in step (2) and column $-j$ if they were), so each move of the form $x_j$ or $x_{-j}$ either flips both cubies or flips neither; the two cubies are flipped an equal number of times. Thus we see that the two cubies are flipped an equal number of times by such a transformation. 

We can conclude that the two cubies are flipped an equal number of times by transformation $m_{k'} \circ\cdots \circ m_1 \circ a_1$. In configuration $C_b$, the two cubies have different colors on their top faces, so after transformation $m_{k'} \circ\cdots \circ m_1 \circ a_1$ flips each of the two cubies an equal number of times, the resulting configuration still has different colors on the top faces of the two cubies. But the resulting configuration is $C'$, which has red as the top face color of every cubie. Thus we have our desired contradiction. Therefore the useful fact is true and the desired result holds.
\end{proof}

\subsection{Conclusion}

Theorems~\ref{thm:square_first_direction} and~\ref{thm:square_second_direction} and Corollaries~\ref{corollary:square_first_direction} and~\ref{corollary:square_second_direction} show that the polynomial-time reductions given are answer preserving. As a result, we conclude that 

\begin{theorem}
The Rubik's Square and Group Rubik's Square problems are NP-complete.
\end{theorem}

\section{(Group) STM/SQTM Rubik's Cube is NP-complete}
\label{section:rubiks_cube}

\subsection{Reductions}

Below, we introduce the reductions used for the Rubik's Cube case. These reductions very closely mirror the Rubik's Square case, and the intuition remains exactly the same: the $b_i$ terms commute, and so if the input Promise Cubical Hamiltonian Path instance is a ``yes'' instance then the $b_i$s can be reordered so that all but $k$ moves in the definition of $t$ will cancel; therefore in that case $t$ can be both enacted and reversed in $k$ moves. 

There are, however, several notable differences from the Rubik's Square case. The first difference is that in a Rubik's Cube, the moves $x_i$, $y_i$, and $z_i$ are all quarter turn rotations rather than self-inverting row or column flips. One consequence is that unlike in the Rubik's Square case, the term $a_i$ does not have the property that $(a_i)^{-1} = a_i$. A second difference is that in a Rubik's Square, the rows never become columns or visa versa. In a Rubik's Cube on the other hand, rotation of the faces can put rows of stickers that were once aligned parallel to one axis into alignment with another axis. To avoid allowing a solution of the puzzle due to this fact in the absence of a solution to the input Promise Cubical Hamiltonian Path instance, the slices in this construction which take the role of rows $1$ through $n$ in the Rubik's Square case and the slices which take the role of columns $1$ through $m$ in the Rubik's Square case will be assigned entirely distinct indices.

To prove that the STM/SQTM Rubik's Cube and Group STM/SQTM Rubik's Cube problems are NP-complete, we reduce from the Promise Cubical Hamiltonian Path problem of Section~\ref{section:promise_cubical_hamiltonian_path} as described below.

Suppose we are given an instance of the Promise Cubical Hamiltonian Path problem consisting of $n$ biststrings $l_1, \ldots, l_n$ of length $m$ (with $l_n = 00\ldots0$). To construct a Group STM/SQTM Rubik's Square instance we need to compute the value $k$ indicating the allowed number of moves and construct the transformation $t$ in $RC_s$.

The value $k$ can be computed directly as $k = 2n - 1$.

The transformation $t$ will be an element of group $RC_s$ where $s = 6n+2m$. Define $a_i$ for $1\le i \le n$ to be $(x_1)^{(l_i)_1}\circ(x_2)^{(l_i)_2}\circ\cdots\circ(x_m)^{(l_i)_m}$ where $(l_i)_1, (l_i)_2, \ldots, (l_i)_m$ are the bits of $l_i$. Also define $b_i = (a_i)^{-1} \circ z_{m+i} \circ a_i$ for $1\le i \le n$. Then we define $t$ to be $a_1\circ b_1 \circ b_2 \circ\cdots\circ b_n$.

Outputting $(t, k)$ completes the reduction from the Promise Cubical Hamiltonian Path problem to the Group STM/SQTM Rubik's Cube problem. To reduce from the Promise Cubical Hamiltonian Path problem to the STM/SQTM Rubik's Cube problem we simply output $(C_t, k) = (t(C_0), k)$. As with the Rubik's Square case, these reductions are clearly polynomial-time reductions.

\subsection{Promise Cubical Hamiltonian Path solution $\to$ (Group) STM/SQTM Rubik's Cube solution}

In this section, we prove one direction of the answer preserving property of the reductions. This proof is not substantively different from the proof of the first direction for the Rubik's Square problems (in Section~\ref{section:square_first_direction}). The differences in these proofs are all minor details that are only present to account for the differences (listed above) between the Rubik's Square and Rubik's Cube reductions.

\begin{lemma}
The transformations $b_i$ all commute.
\label{lemma:cube_b_i_commute}
\end{lemma}

\begin{proof}
Consider any such transformation $b_i$. The transformation $b_i$ can be rewritten as $(a_i)^{-1} \circ z_{m+i} \circ a_i$. For any cubie not moved by the $z_{m+i}$ middle term, the effect of this transformation is the same as the effect of transformation $(a_i)^{-1} \circ a_i = 1$. In other words, $b_i$ only affects cubies that are moved by the $z_{m+i}$ term. 

A cubie affected by this term was either moved into the $z$ slice with index $(m+i)$ by $a_i$ or was already there. $a_i$ consists of some number of clockwise $x$ turns. Thus, in order to be moved into a position with $z = (m+i)$, a cubie would have to start in a position with $y = -(m+i)$ on the $+z$ face or in a position with $y = (m+i)$ on the $-z$ face.

Thus, the cubies affected by $b_i$ must either have $y$ coordinate $\pm(m+i)$ and lie on one of the $\pm z$ faces or have $z$ coordinate $(m+i)$ and lie on one of the other four faces. This is enough to see that the cubies affected by $b_i$ are disjoint from those affected by $b_j$ (for $j \ne i$). In other words, the transformations $b_i$ all commute.
\end{proof}

\begin{theorem}
If $l_1, \ldots, l_n$ is a ``yes'' instance to the Promise Cubical Hamiltonian Path problem, then $(t, k)$ is a ``yes'' instance to the Group SQTM Rubik's Cube problem.
\label{thm:cube_first_direction}
\end{theorem}

\begin{proof}
Suppose $l_1, \ldots, l_n$ is a ``yes'' instance to the Promise Cubical Hamiltonian Path problem. Let $m$ be the length of $l_i$ and note that $l_n = 00\ldots0$ by the promise of the Promise Cubical Hamiltonian Path problem. Furthermore, since $l_1, \ldots, l_n$ is a ``yes'' instance to the Promise Cubical Hamiltonian Path problem, there exists an ordering of these bitstrings $l_{i_1}, l_{i_2}, \ldots, l_{i_n}$ such that each consecutive pair of bitstrings is at Hamming distance one, $i_1 = 1$, and $i_n = n$ (with the final two conditions coming from the promise).

By Lemma~\ref{lemma:cube_b_i_commute}, we know that $t = a_1 \circ b_1 \circ b_2 \circ\cdots \circ b_n$ can be rewritten as 
$$t = a_1 \circ b_{i_1} \circ b_{i_2} \circ\cdots \circ b_{i_n}.$$ 
Using the definition of $b_i$, we can further rewrite this as 
$$t = a_1 \circ ((a_{i_1})^{-1} \circ z_{m+i_1} \circ a_{i_1}) \circ ((a_{i_2})^{-1} \circ z_{m+i_2} \circ a_{i_2}) \circ\cdots \circ ((a_{i_n})^{-1} \circ z_{m+i_n} \circ a_{i_n}),$$
or as 
$$t = (a_1 \circ (a_{i_1})^{-1}) \circ z_{m+i_1} \circ (a_{i_1} \circ (a_{i_2})^{-1}) \circ z_{m+i_2} \circ (a_{i_2} \circ (a_{i_3})^{-1}) \circ\cdots \circ (a_{i_{n-1}} \circ (a_{i_n})^{-1}) \circ z_{m+i_n} \circ (a_{i_n}).$$

We know that $i_1 = 1$, and therefore that $a_1 \circ (a_{i_1})^{-1} = a_1 \circ (a_1)^{-1} = 1$ is the identity element. Similarly, we know that $i_n = n$ and therefore that $a_{i_n} = a_n = (x_1)^{(l_n)_1}\circ(x_2)^{(l_n)_2}\circ\cdots\circ(x_m)^{(l_n)_m} =  (x_1)^{0}\circ(x_2)^{0}\circ\cdots\circ(x_m)^{0} = 1$ is also the identity.

Thus we see that $$t = z_{m+i_1} \circ (a_{i_1} \circ (a_{i_2})^{-1}) \circ z_{m+i_2} \circ (a_{i_2} \circ (a_{i_3})^{-1}) \circ\cdots \circ (a_{i_{n-1}} \circ (a_{i_n})^{-1}) \circ z_{m+i_n}.$$

Consider the transformation $a_{i_p} \circ (a_{i_{p+1}})^{-1}$. This transformation can be written as $$a_{i_p} \circ (a_{i_{p+1}})^{-1} = (x_1)^{(l_{i_p})_1} \circ (x_2)^{(l_{i_p})_2} \circ\cdots \circ (x_m)^{(l_{i_p})_m} \circ (x_1)^{-(l_{i_{p+1}})_1} \circ (x_2)^{-(l_{i_{p+1}})_2} \circ\cdots \circ (x_m)^{-(l_{i_{p+1}})_m}.$$

Because $x_u$ always commutes with $x_v$, we can rewrite this as $$a_{i_p} \circ (a_{i_{p+1}})^{-1} = (x_1)^{(l_{i_p})_1 - (l_{i_{p+1}})_1} \circ (x_2)^{(l_{i_p})_2-(l_{i_{p+1}})_2} \circ\cdots \circ (x_m)^{(l_{i_p})_m-(l_{i_{p+1}})_m}.$$

Since $l_{i_p}$ differs from $l_{i_{p+1}}$ in only one position, call it $j_p$, we see that $(l_{i_p})_j-(l_{i_{p+1}})_j$ is zero unless $j = j_p$, and is $\pm 1$ in that final case. This is sufficient to show that $a_{i_p} \circ (a_{i_{p+1}})^{-1} = (x_{j_p})^{s_p}$ where $s_p = \pm 1$.

Thus we see that $$t = z_{m+i_1} \circ (x_{j_1})^{s_1} \circ z_{m+i_2} \circ (x_{j_2})^{s_2} \circ\cdots \circ (x_{j_{n-1}})^{s_{n-1}} \circ z_{m+i_n},$$
or (by left multiplying) that 
$$(z_{m+i_n})^{-1} \circ (x_{j_{n-1}})^{-s_{n-1}} \circ\cdots \circ (x_{j_2})^{-s_2} \circ (z_{m+i_2})^{-1} \circ (x_{j_1})^{-s_1} \circ (z_{m+i_1})^{-1} \circ t = 1.$$ 
We see that $t$ can be reversed by $k = 2n-1$ terms of the form $(z_i)^{-1}$, $x_j$, and $(x_j)^{-1}$, which are all SQTM moves. In other words, $(t, k)$ is a ``yes'' instance to the Group SQTM Rubik's Cube problem.
\end{proof}

\begin{corollary}
If $l_1, \ldots, l_n$ is a ``yes'' instance to the Promise Cubical Hamiltonian Path problem, then $(C_t, k)$ is a ``yes'' instance to the STM/SQTM Rubik's Cube problem and $(t, k)$ is a ``yes'' instance to the Group STM/SQTM Rubik's Cube problem.
\label{corollary:cube_first_direction}
\end{corollary}

\begin{proof}
This follows immediately from Theorem~\ref{thm:cube_first_direction} and Lemmas~\ref{lemma:cube_types_1} and~\ref{lemma:cube_types_2}.
\end{proof}

\subsection{\boldmath Coloring of $C_t$}

As in the Rubik's Square case, it will be helpful for the second direction of the proof to know the coloring of the Cube's configuration. As before, we define $b = b_1 \circ \cdots \circ b_n$ (so that $t = a_1 \circ b$) and determine the colors of the stickers in configuration $C_b = b(C_0)$.

Consider the example instance of Promise Cubical Hamiltonian Path with $n = 5$ and $m = 3$ introduced in the Rubik's Square section and reproduced below:

\begin{align*}
l_1 &= 011 \\
l_2 &= 110 \\
l_3 &= 111 \\
l_4 &= 100 \\
l_5 &= 000
\end{align*}

For this example instance, the Rubik's Cube configuration produced by the reduction is an $s \times s \times s$ Rubik's Cube with $s = 2m+6n = 36$. Furthermore, the coloring of the stickers in $C_b$ for this example is shown in Figure~\ref{fig:cube_example_2}. Note that the $n\times m$ grid of bits comprising $l_1, \ldots, l_n$ is actually directly encoded in the coloring of each face.

\begin{figure}[h]
\centering
\includegraphics[width=\textwidth]{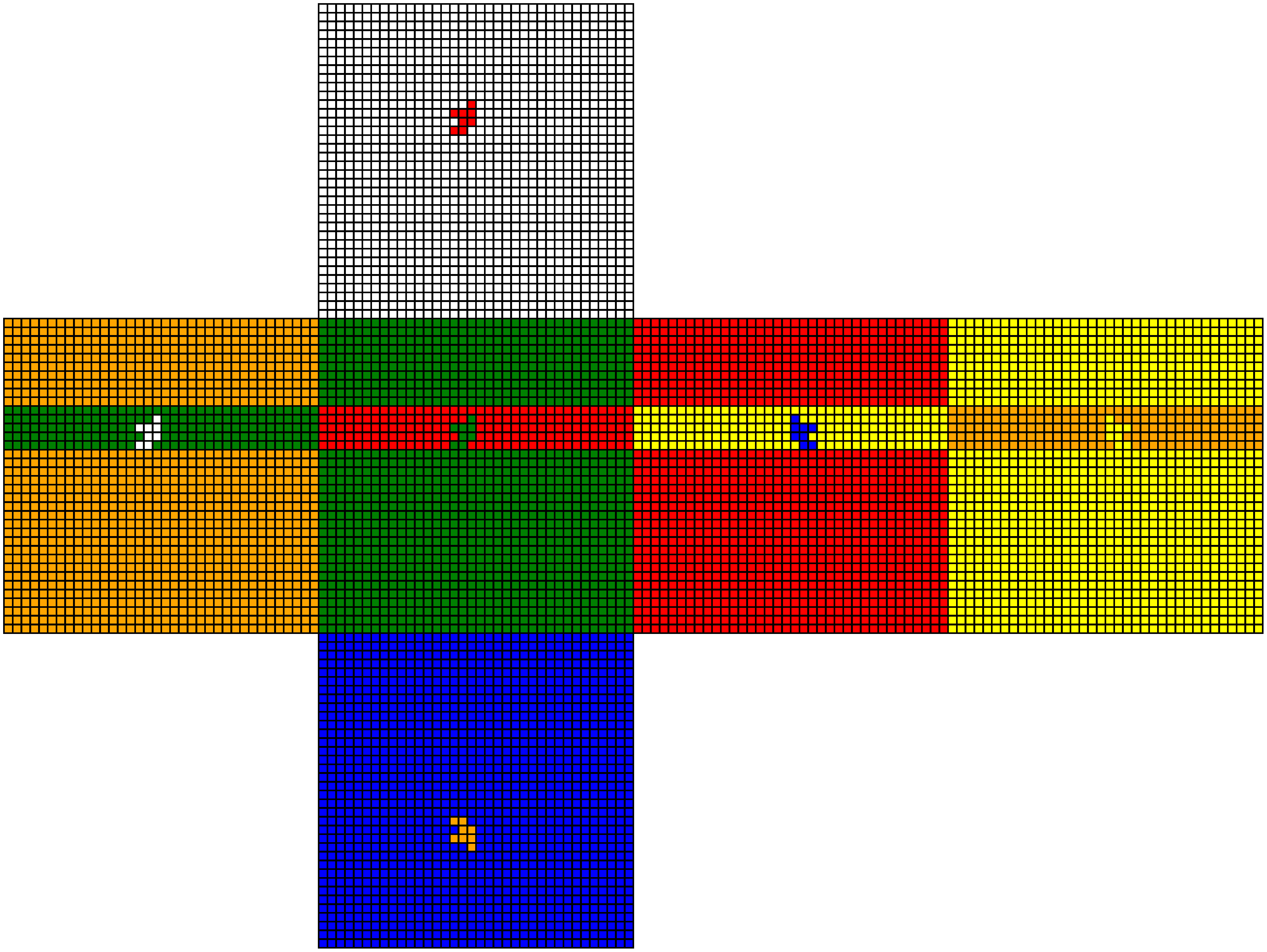}
\caption{The faces of $C_b$ for the example input $l_1, \ldots, l_n$. In this figure, the top and bottom faces are the $+z$ and $-z$ faces, while the faces in the vertical center of the figure are the $+x$, $+y$, $-x$, and $-y$ faces from left to right.}
\label{fig:cube_example_2}
\end{figure}

In this section, we prove the following useful theorem, which formalizes the pattern of colors from the example (Figure~\ref{fig:cube_example_2}):

\begin{restatable}{theorem}{cubecoloringthm}
\label{thm:cube_coloring}
In $C_b$, the stickers have the following coloring:
\begin{enumerate}
\item[$+z$:] The stickers on the $+z$ face with $(x, y)$ coordinates $(j, -(m+i))$ where $i \in \{1, \ldots, n\}$ and the $j$th bit of $l_i$ is one are all red. All other stickers are white.
\item[$-z$:] The stickers on the $-z$ face with $(x, y)$ coordinates $(j, -(m+i))$ where $i \in \{1, \ldots, n\}$ and the $j$th bit of $l_i$ is one are all orange. All other stickers are blue. 
\item[$+y$:] The stickers on the $+y$ face with $(x, z)$ coordinates $(j, (m+i))$ where $i \in \{1, \ldots, n\}$ and either $l_i$ doesn't have a $j$th bit (i.e. $j < 0$ or $j > m$) or the $j$th bit of $l_i$ is zero are all red. All other stickers are green.
\item[$-y$:] The stickers on the $-y$ face with $(x, z)$ coordinates $(j, (m+i))$ where $i \in \{1, \ldots, n\}$ and either $l_i$ doesn't have a $j$th bit (i.e. $j < 0$ or $j > m$) or the $j$th bit of $l_i$ is zero are all orange. All other stickers are yellow.
\item[$+x$:] The stickers on the $+x$ face with $(y, z)$ coordinates $(-j, (m+i))$ where $i \in \{1, \ldots, n\}$ and the $j$th bit of $l_i$ is one are all white. All other stickers with $z$ coordinate in $\{1, \ldots, n\}$ are green. All other stickers are orange.
\item[$-x$:] The stickers on the $-x$ face with $(y, z)$ coordinates $(-j, (m+i))$ where $i \in \{1, \ldots, n\}$ and the $j$th bit of $l_i$ is one are all blue. All other stickers with $z$ coordinate in $\{1, \ldots, n\}$ are yellow. All other stickers are red.
\end{enumerate}
\end{restatable}

The proof of this theorem is involved and uninsightful. In addition, no other result from this section will be used in the rest of this paper. As a result, the reader should feel free to skip the remainder of this section.

To formally derive the coloring of configuration $C_b$, we need to have a formal description of the effect of transformation $b_i$. For example, Figure~\ref{fig:cube_example_1} shows the $+x$, $+y$, and $+z$ faces of a Rubik's Cube in configurations $C_0$, $a_2(C_0)$, $(z_{m+2} \circ a_2)(C_0)$, and $b_2(C_0) = ((a_2)^{-1} \circ z_{m+2} \circ a_2)(C_0)$ where $a_2$ and $z_{m+2} = z_5$ are defined in terms of $l_2 = 110$ as in the reduction.

\begin{figure}[h]
\centering
\begin{subfigure}[b]{0.48\textwidth}
\includegraphics[width=\textwidth]{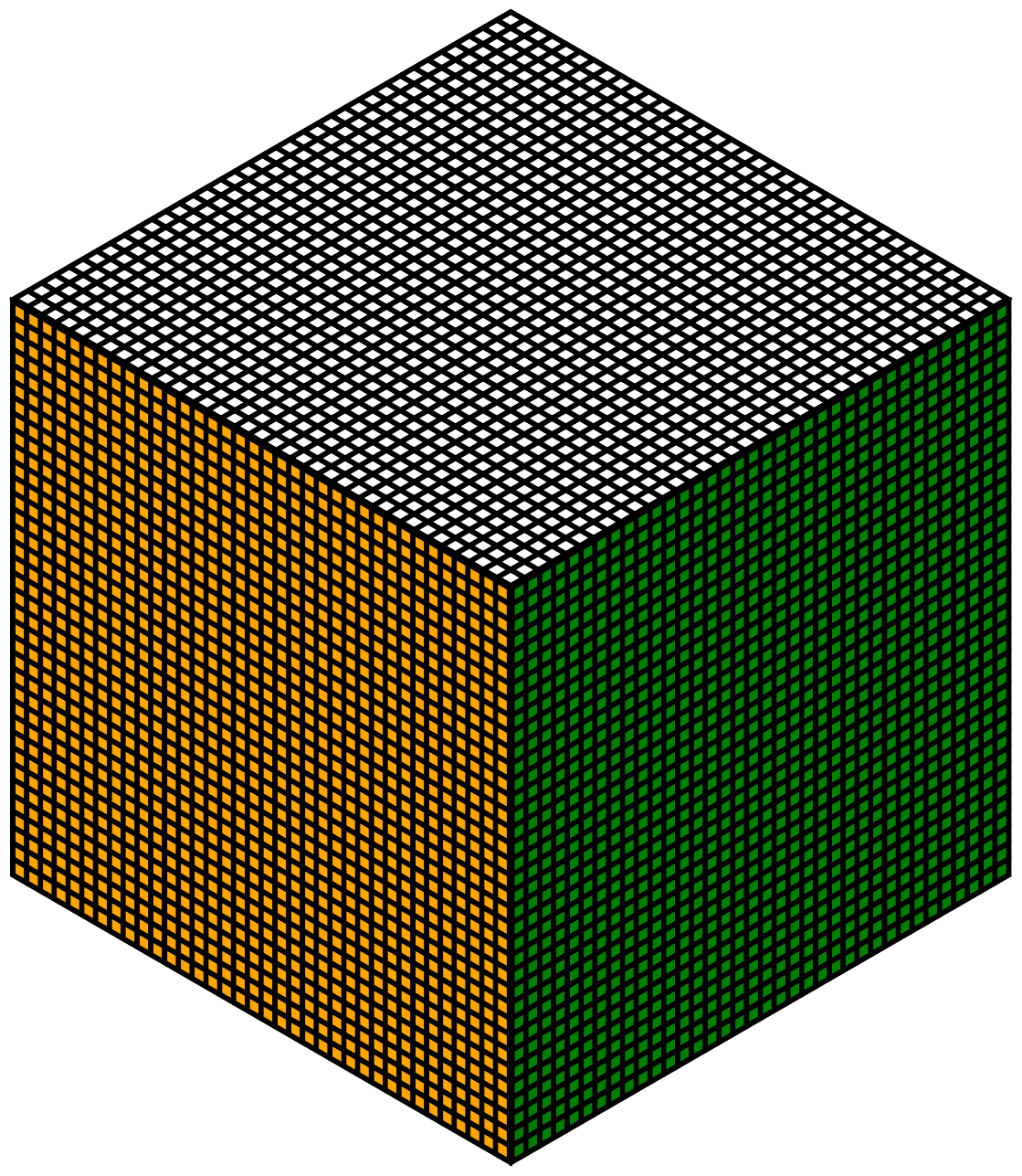}
\caption{\label{fig:cube_example_1a}}
\end{subfigure}
\hfill
\begin{subfigure}[b]{0.48\textwidth}
\includegraphics[width=\textwidth]{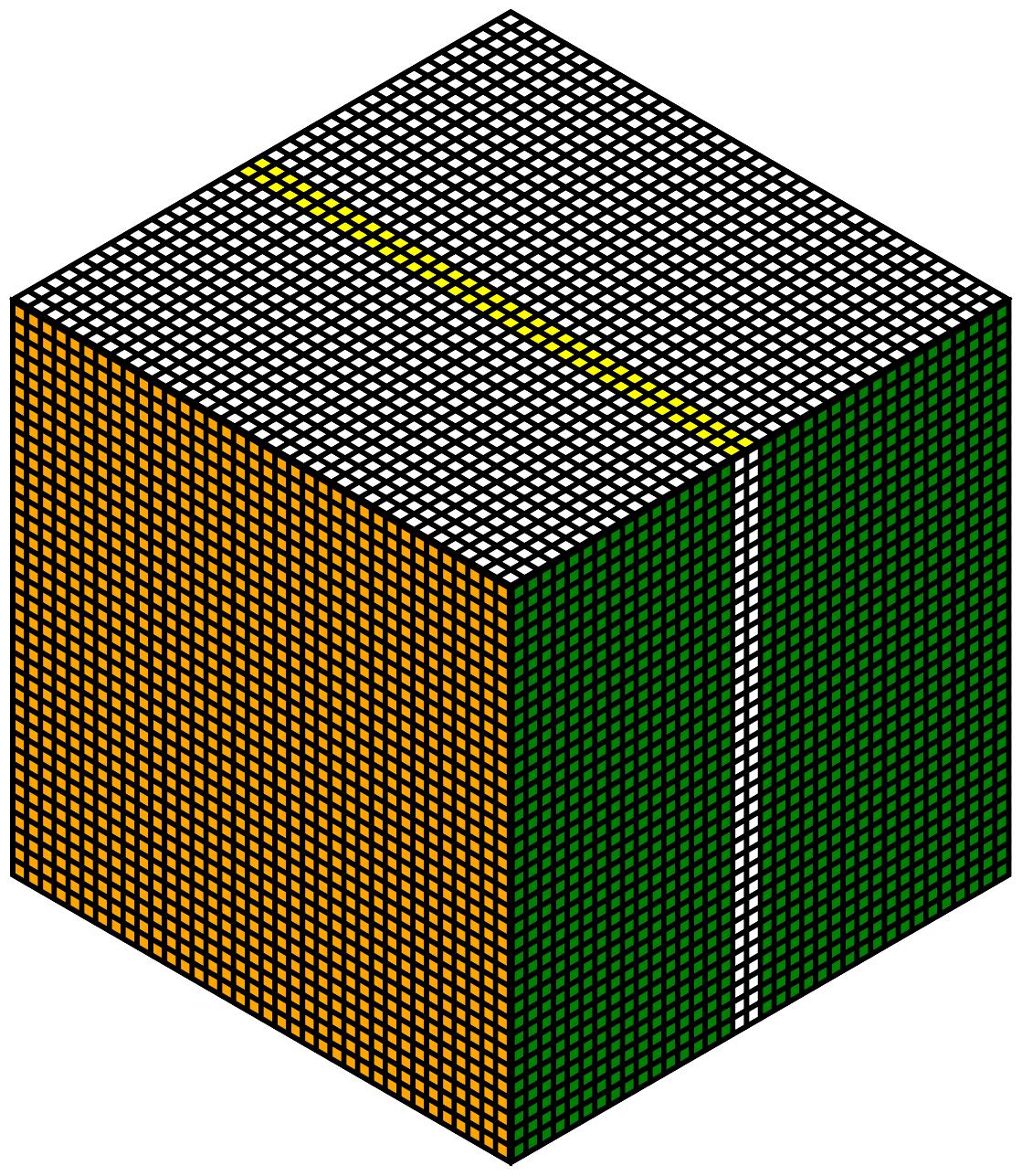}
\caption{\label{fig:cube_example_1b}}
\end{subfigure}

\begin{subfigure}[b]{0.48\textwidth}
\includegraphics[width=\textwidth]{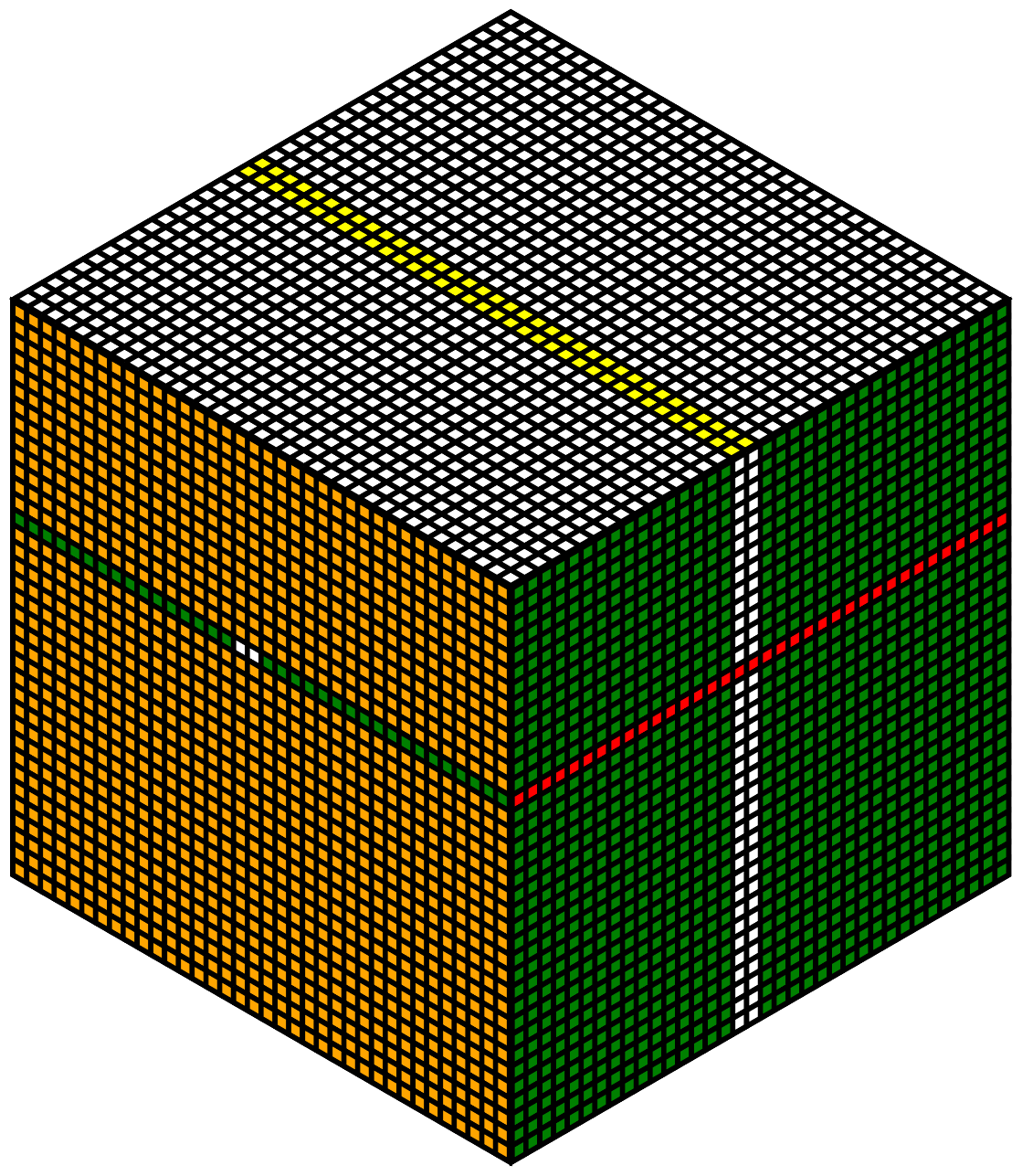}
\caption{\label{fig:cube_example_1c}}
\end{subfigure}
\hfill
\begin{subfigure}[b]{0.48\textwidth}
\includegraphics[width=\textwidth]{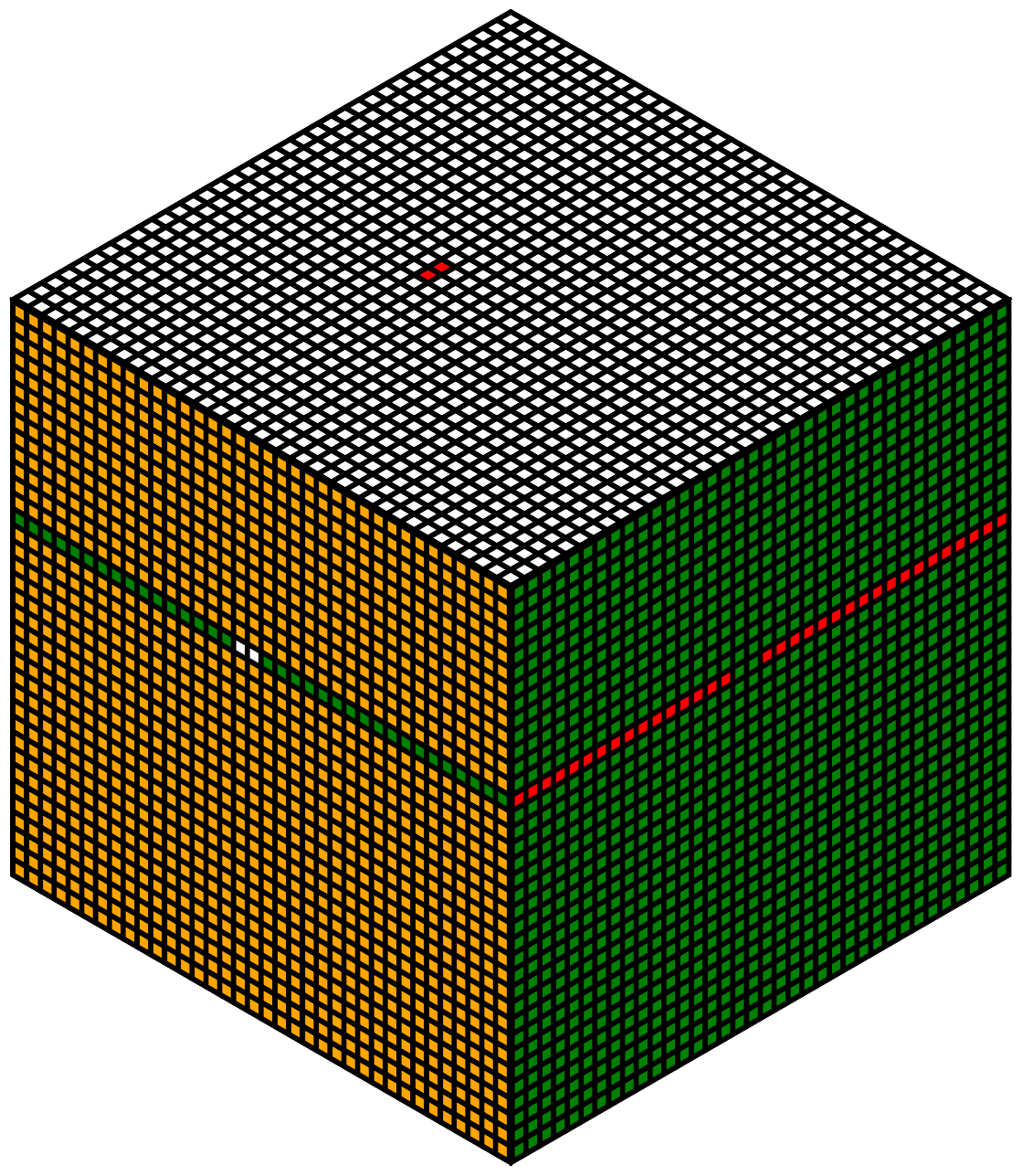}
\caption{\label{fig:cube_example_1d}}
\end{subfigure}
\caption{Applying $b_2$ to $C_0$ step by step.}
\label{fig:cube_example_1}
\end{figure}

The exact behavior of a Rubik's Cube due to $b_i$ is described by Lemmas~\ref{lemma:cube_b_i_effect_1} through~\ref{lemma:cube_b_i_effect_3}:

\begin{lemma}
\label{lemma:cube_b_i_effect_1}
Suppose $i \in \{1,\ldots,n\}$. Then the effect of $b_i$ on the stickers from the $\pm z$ faces of a Rubik's Cube can be described as follows:
\begin{itemize}
\item If the $j$th bit of $l_i$ is one, then the sticker starting on the $+z$ face with $(x, y)$ coordinates $(j, -(m+i))$ ends up on the $+x$ face with $(y, z)$ coordinates $(-j, (m+i))$.

\item If the $j$th bit of $l_i$ is one, then the sticker starting on the $-z$ face with $(x, y)$ coordinates $(j, -(m+i))$ ends up on the $-x$ face with $(y, z)$ coordinates $(-j, (m+i))$.

\item All other stickers on the $\pm z$ faces stay in place.
\end{itemize}
\end{lemma}

\begin{proof}
As noted in the proof of Lemma~\ref{lemma:cube_b_i_commute}, a sticker is affected by $b_i = (a_i)^{-1} \circ z_{m+i} \circ a_i$ if and only if it is moved by the $z_{m+i}$ term.

Consider the stickers originally on the $+z$ face. $b_i$ starts with $a_i$, which rotates the $x$ slices with $x$ coordinates $j$ such that bit $j$ of $l_i$ is one. Therefore, the stickers on the $+z$ face with $x$ coordinates of this form are rotated to the $+y$ face, and all the other stickers are left in place. After that, the only stickers from the $+z$ face which are moved by the $z_{m+i}$ term of $b_i$ are the stickers which were on the $+y$ face with $z$ coordinate $(m+i)$ and $x$ coordinate $j$ such that bit $j$ of $l_i$ is one. In other words, the only stickers from the $+z$ face moved by the $z_{m+i}$ term, are those starting at $(x,y)$ coordinates $(j, -(m+i))$ where bit $j$ of $l_i$ is one. 

All other stickers starting on the $+z$ face are not affected by the $z_{m+i}$ term, and are therefore not moved by $b_i$. On the other hand, consider any sticker of this form: a sticker starting on the $+z$ face at $(x,y)$ coordinates $(j, -(m+i))$ where bit $j$ of $l_i$ is one. Such a sticker is moved by $a_i$ to $(x, z)$ coordinates $(j, (m+i))$ of face $+y$. It is then moved by $z_{m+i}$ to $(y, z)$ coordinates $(-j, (m+i))$ of face $+x$. Finally, $(a_i)^{-1}$ does not affect the sticker since it is on the $+x$ face at the time and $(a_i)^{-1}$ consists of rotations of $x$ slices.

Thus, if the $j$th bit of $l_i$ is one, then the sticker starting on the $+z$ face with $(x, y)$ coordinates $(j, -(m+i))$ ends up on the $+x$ face with $(y, z)$ coordinates $(-j, (m+i))$. All other stickers starting on the $+z$ face remain in place.

The exact same logic applies to the stickers originally on the $-z$ face, allowing us to conclude that the lemma statement holds, as desired.
\end{proof}

\begin{lemma}
\label{lemma:cube_b_i_effect_2}
Suppose $i \in \{1,\ldots,n\}$. Then the effect of $b_i$ on the stickers from the $\pm y$ faces of a Rubik's Cube can be described as follows:
\begin{itemize}
\item If the $j$th bit of $l_i$ does not exist (i.e. $j < 0$ or $j > m$) or if the $j$th bit of $l_i$ is zero, then the sticker starting on the $+y$ face with $(x, z)$ coordinates $(j, (m+i))$ ends up on the $+x$ face with $(y, z)$ coordinates $(-j, (m+i))$. 

\item If the $j$th bit of $l_i$ does not exist (i.e. $j < 0$ or $j > m$) or if the $j$th bit of $l_i$ is zero, then the sticker starting on the $-y$ face with $(x, z)$ coordinates $(j, (m+i))$ ends up on the $-x$ face with $(y, z)$ coordinates $(-j, (m+i))$. 

\item All other stickers on the $\pm y$ faces stay in place.
\end{itemize}
\end{lemma}

\begin{proof}
As noted in the proof of Lemma~\ref{lemma:cube_b_i_commute}, a sticker is affected by $b_i = (a_i)^{-1} \circ z_{m+i} \circ a_i$ if and only if it is moved by the $z_{m+i}$ term.

Consider the stickers originally on the $+y$ face. $b_i$ starts with $a_i$, which rotates the $x$ slices with $x$ coordinates $j$ such that bit $j$ of $l_i$ is one. Therefore, the stickers on the $+y$ face with $x$ coordinates of this form are rotated to the $-z$ face, and all the other stickers are left in place. After that, the only stickers from the $+y$ face which are moved by the $z_{m+i}$ term of $b_i$ are the stickers with $z$ coordinate $(m+i)$ which were not moved from the $+y$ face. In other words, the only stickers from the $+z$ face moved by the $z_{m+i}$ term, are those starting at $(x,y)$ coordinates $(j, -(m+i))$ where bit $j$ of $l_i$ either does not exist (i.e. $j < 0$ or $j > m$) or is zero. 

All other stickers starting on the $+y$ face are not affected by the $z_{m+i}$ term, and are therefore not moved by $b_i$. On the other hand, consider any sticker of this form: a sticker starting on the $+y$ face at $(x,z)$ coordinates $(j, (m+i))$ where bit $j$ of $l_i$ either does not exist or is zero. Such a sticker is not moved by $a_i$. It is then moved by $z_{m+i}$ to $(y, z)$ coordinates $(-j, (m+i))$ of face $+x$. Finally, $(a_i)^{-1}$ does not affect the sticker since it is on the $+x$ face at the time and $(a_i)^{-1}$ consists of rotations of $x$ slices.

Thus, if the $j$th bit of $l_i$ does not exist (i.e. $j < 0$ or $j > m$) or if the $j$th bit of $l_i$ is zero, then the sticker starting on the $+y$ face with $(x, z)$ coordinates $(j, (m+i))$ ends up on the $+x$ face with $(y, z)$ coordinates $(-j, (m+i))$.  All other stickers starting on the $+y$ face remain in place.

The exact same logic applies to the stickers originally on the $-y$ face, allowing us to conclude that the lemma statement holds, as desired.
\end{proof}

\begin{lemma}
\label{lemma:cube_b_i_effect_3}
Suppose $i \in \{1,\ldots,n\}$. Then the effect of $b_i$ on the stickers from the $\pm x$ faces of a Rubik's Cube can be described as follows:
\begin{itemize}
\item If the $j$th bit of $l_i$ is one, then the sticker starting on the $+x$ face with $(y, z)$ coordinates $(j, (m+i))$ ends up on the $-z$ face with $(x, y)$ coordinates $(j, -(m+i))$.

\item If the $j$th bit of $l_i$ does not exist (i.e. $j < 0$ or $j > m$) or if the $j$th bit of $l_i$ is zero, then the sticker starting on the $+x$ face with $(y, z)$ coordinates $(j, (m+i))$ ends up on the $-y$ face with $(x, z)$ coordinates $(j, (m+i))$. 

\item If the $j$th bit of $l_i$ is one, then the sticker starting on the $-x$ face with $(y, z)$ coordinates $(j, (m+i))$ ends up on the $+z$ face with $(x, y)$ coordinates $(j, -(m+i))$.

\item If the $j$th bit of $l_i$ does not exist (i.e. $j < 0$ or $j > m$) or if the $j$th bit of $l_i$ is zero, then the sticker starting on the $-x$ face with $(y, z)$ coordinates $(j, (m+i))$ ends up on the $+y$ face with $(x, z)$ coordinates $(j, (m+i))$. 

\item All other stickers on the $\pm x$ faces stay in place.
\end{itemize}
\end{lemma}

\begin{proof}
As noted in the proof of Lemma~\ref{lemma:cube_b_i_commute}, a sticker is affected by $b_i = (a_i)^{-1} \circ z_{m+i} \circ a_i$ if and only if it is moved by the $z_{m+i}$ term.

Consider the stickers originally on the $+x$ face. $b_i$ starts with $a_i$, which affects none of the stickers on the $+x$ face. After that, the $z_{m+i}$ term moves exactly those stickers from the $+x$ face that had $z$ coordinate $(m+i)$. As a result, these stickers are all affected by $b_i$, and all others are not.

Consider a sticker starting on the $+x$ face with $(y, z)$ coordinates $(j, (m+i))$. This sticker is unaffected by $a_i$ and then moved to the $-y$ face by $z_{m+i}$. In particular, it is moved to $(x, z)$ coordinates $(j, (m+i))$. After that, there are two cases:

Case 1: If bit $j$ of $l_i$ does not exist (i.e. $j < 0$ or $j > m$) or if the $j$th bit of $l_i$ is zero, then the sticker is unaffected by $(a_i)^{-1}$. This shows that if the $j$th bit of $l_i$ does not exist or if the $j$th bit of $l_i$ is zero, then the sticker starting on the $+x$ face with $(y, z)$ coordinates $(j, (m+i))$ ends up on the $-y$ face with $(x, z)$ coordinates $(j, (m+i))$. 

Case 2: If bit $j$ of $l_i$ is one, then after being moved to the $-y$ face by $z_{m+i}$, the sticker in question is moved to the $-z$ face by $(a_i)^{-1}$. In particular, the sticker ends up at $(x, y)$ coordinates $(j, -(m+i))$. This shows that if the $j$th bit of $l_i$ is one, then the sticker starting on the $+x$ face with $(y, z)$ coordinates $(j, (m+i))$ ends up on the $-z$ face with $(x, y)$ coordinates $(j, -(m+i))$.

As previously mentioned, all stickers starting on the $+x$ face other than those addressed by the above cases stay in place due to $b_i$. Together with the statements shown in the two cases, this is exactly what we wished to show.

The same logic applies to the stickers originally on the $-x$ face, allowing us to conclude that the lemma statement holds, as desired.
\end{proof}

We can apply the above lemmas to figure out the effect of transformation $b_1 \circ b_2 \circ \cdots \circ b_n$ on configuration $C_0$. In particular, this allows us to learn the coloring of configuration $C_b$.

At this point, we can prove Theorem~\ref{thm:cube_coloring}, which is restated below for convenience:

\cubecoloringthm*

\begin{proof}
$C_b$ is obtained from $C_0$ by applying transformation $b_1 \circ b_2 \circ \cdots \circ b_n$. Each $b_i$ affects a disjoint set of stickers. Using this fact together with the description of the effect of one $b_i$, we can obtain the description of the coloring of $C_b$ given in the above theorem statement.

For example, consider the stickers that end up on the $+z$ face. According to Lemma~\ref{lemma:cube_b_i_effect_3}, if the $j$th bit of $l_i$ is one, then $b_i$ moves the sticker starting on the $-x$ face with $(y, z)$ coordinates $(-j, (m+i))$ to the $+z$ face with $(x, y)$ coordinates $(j, -(m+i))$. Since the $b_i$s each affect disjoint sets of stickers, the stickers on the $+z$ face with $(x, y)$ coordinates $(j, -(m+i))$ where $i \in \{1, \ldots, n\}$ and the $j$th bit of $l_i$ is one are all stickers that started on the $-x$ face. Since the $-x$ face is red in $C_0$, these stickers are all red. We know from Lemmas~\ref{lemma:cube_b_i_effect_1} through~\ref{lemma:cube_b_i_effect_3} that no stickers other than the ones that started there and those described above are moved to the $+z$ face by $b_i$. Therefore all other stickers on the $+z$ face started there. Since the $+z$ face is white in $C_0$, these stickers are all white. Putting this together, we obtain exactly the first bullet point of the theorem statement:

The stickers on the $+z$ face with $(x, y)$ coordinates $(j, -(m+i))$ where $i \in \{1, \ldots, n\}$ and the $j$th bit of $l_i$ is one are all red. All the other stickers are white.

The logic for the other five faces is exactly analogous, and is omitted here for brevity.
\end{proof}

This concludes the description of $C_b$ in terms of colors. The coloring of configuration $C_t$---the configuration that is actually obtained by applying the reduction to $l_1, \ldots, l_n$---can be obtained from the coloring of configuration $C_b$ by applying transformation $a_1$. This is shown for the previously given example in Figure~\ref{fig:cube_example_3}.

\begin{figure}[h]
\centering
\includegraphics[width=0.5\textwidth]{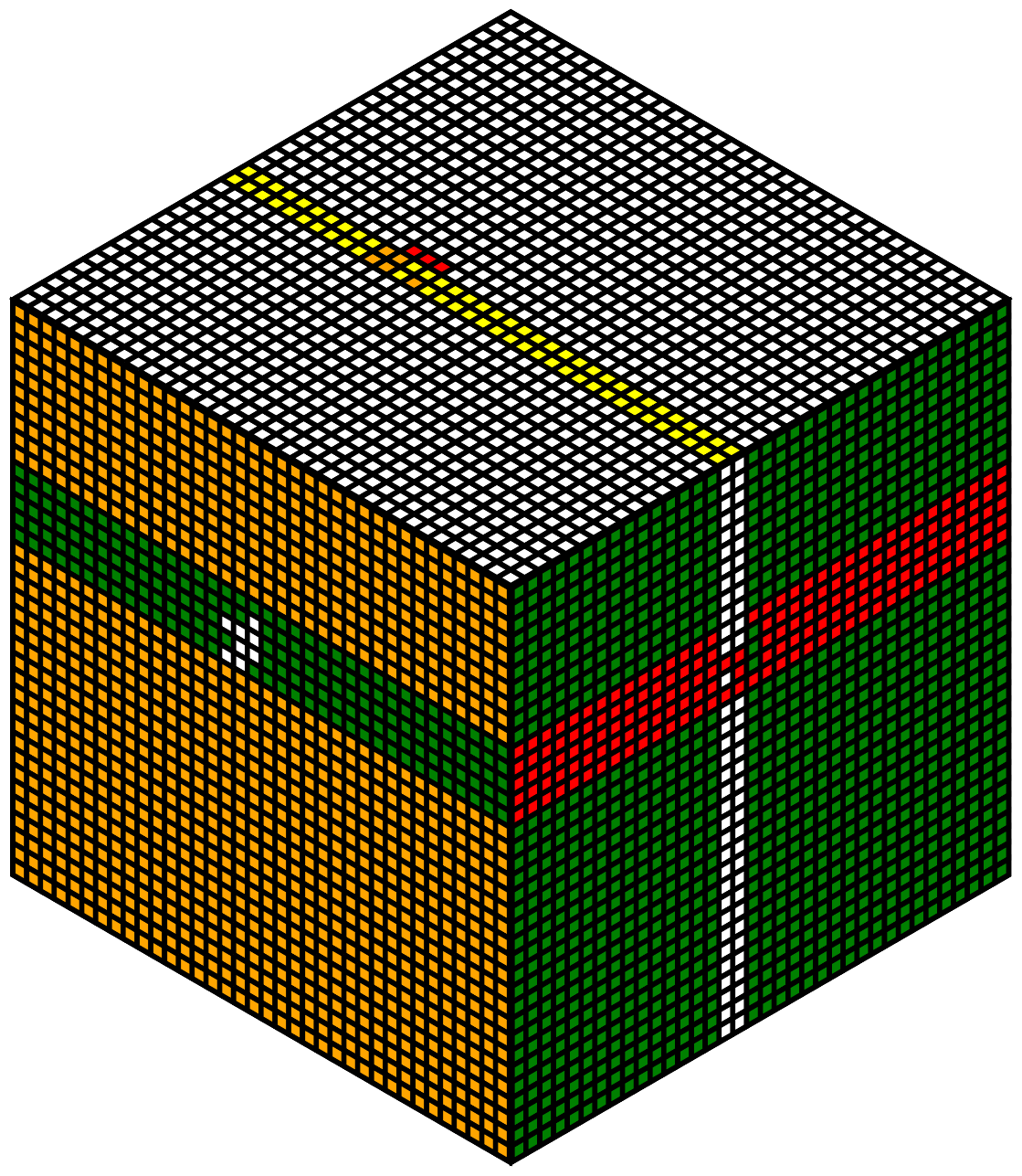}
\caption{The $+x$, $+y$, and $+z$ faces of $C_t$ for the example input $l_1, \ldots, l_n$.}
\label{fig:cube_example_3}
\end{figure}

\subsection{(Group) STM/SQTM Rubik's Cube solution $\to$ Promise Cubical Hamiltonian Path solution: proof outline}
\label{section:cube_second_direction}

We wish to prove the following:

\begin{theorem}
\label{thm:cube_second_direction}
If $(C_t, k)$ is a ``yes'' instance to the STM Rubik's Cube problem, then $l_1, \ldots, l_n$ is a ``yes'' instance to the Promise Cubical Hamiltonian Path problem.
\end{theorem}

By Lemmas~\ref{lemma:cube_types_1} and~\ref{lemma:cube_types_2}, this will immediately also imply the following corollary:

\begin{corollary}
If $(t, k)$ is a ``yes'' instance to the Group STM/SQTM Rubik's Cube problem or $(C_t, k)$ is a ``yes'' instance to the STM/SQTM Rubik's Cube problem, then $l_1, \ldots, l_n$ is a ``yes'' instance to the Promise Cubical Hamiltonian Path problem.
\label{corollary:cube_second_direction}
\end{corollary}

The intuition behind the proof of this theorem is similar to that used in the Rubik's Square case, but there is added complexity due to the extra options available in a Rubik's Cube. Most of the added complexity is due to the possibility of face moves (allowing rows of stickers to align in several directions over the course of a solution).

Below, we describe an outline of the proof, including several high-level steps, each of which is described in more detail in an additional subsection.

To prove the theorem, we consider a hypothetical solution to the $(C_t, k)$ instance of the STM Rubik's Cube problem. A solution consists of a sequence of STM Rubik's Cube moves $m_1, \ldots, m_{k'}$ with $k' \le k$ such that $C' = (m_{k'} \circ \cdots \circ m_1)(C_t)$ is a solved configuration of the Rubik's Cube.

One very helpful idea that is used several times throughout the proof is the idea of an index $u$ such that no move $m_i$ is an index-$u$ move. 

\begin{definition}
Define $u \in \{m+n+1, m+n+2, \ldots, m+n+(2n)\}$ to be an index such that $m_1, \ldots, m_{k'}$ contains no index-$u$ move. 
\end{definition}

Notice that a value for $u$ satisfying this definition must exist because variable $u$ has $2n = k + 1 > k \ge k'$ possible values and each of the $k'$ moves $m_i$ disqualifies at most one possible value from being assigned to variable $u$.

Step 1 of the proof is a preliminary characterization of the possible index-$(m+i)$ moves among $m_1, \ldots, m_{k'}$ for $i \in \{1, \ldots, n\}$. Consider the following definition:

\begin{definition}
Partition the set $\{1, \ldots, n\}$ into four sets of indices $Z$, $O$, $T$, and $M$ (where $Z$, $O$, $T$, and $M$ are named after ``zero'', ``one'', ``two'', and ``more'') as follows:
\begin{itemize}
  \setlength{\itemsep}{1pt}
  \setlength{\parskip}{0pt}
  \setlength{\parsep}{0pt}
\item $i \in Z$ if and only if $m_1, \ldots, m_{k'}$ contains exactly zero index-$(m+i)$ moves
\item $i \in O$ if and only if $m_1, \ldots, m_{k'}$ contains exactly one index-$(m+i)$ move
\item $i \in T$ if and only if $m_1, \ldots, m_{k'}$ contains exactly two index-$(m+i)$ moves
\item $i \in M$ if and only if $m_1, \ldots, m_{k'}$ contains at least three index-$(m+i)$ moves
\end{itemize}
\end{definition}

In Step 1, we prove the following list of results, thereby restricting the set of possible index-$(m+i)$ moves (for $i \in \{1, \ldots, n\}$) among $m_1, \ldots, m_{k'}$:

\begin{itemize}
\item $Z$ is empty.
\item If $i \in O$, then the sole index-$(m+i)$ move must be a counterclockwise $z$ turn.
\item If $i \in T$, then the two index-$(m+i)$ moves must be a clockwise $z$ turn and a $z$ flip in some order.
\item If $i \in O \cup T$, then any move of $z$ slice $(m+i)$ must occur at a time when faces $+x$, $+y$, $-x$, and $-y$ all have zero rotation and any move of $z$ slice $-(m+i)$ must occur at a time when these faces all have rotation $180^\circ$.
\end{itemize}

Step 2 of the proof concerns the concept of paired stickers:

\begin{definition}
Suppose $p_1$, $p_2$, and $q$ are all distinct positive non-face slice indices. Then we say that two stickers are \emph{$(p_1, p_2, q)$-paired} if the two stickers are on the same index-$j$ slice, the two stickers are on the same quadrant of a face, one of the stickers has coordinates $\pm q$ and $\pm p_1$ within that face, and the second sticker has coordinates $\pm q$ and $\pm p_2$ within the face.  
\end{definition}

In particular, we prove the following useful properties of paired stickers:

\begin{itemize}
\item If two stickers are $(p_1, p_2, q)$-paired, then they remain $(p_1, p_2, q)$-paired after one move unless the move is an index-$p_1$ move or an index-$p_2$ move which moves one of the stickers.
\item Suppose $i_1, i_2 \in O$ and $j \in \{1, 2, \ldots, m\}$. Then consider any pair of stickers that are $(m+i_1, m+i_2, j)$-paired in $C_b$. If there are no face moves of faces $+x$, $+y$, $-x$, and $-y$ and no index-$j$ moves that affect either of the stickers between the index-$(m+i_1)$ $O$-move and the index-$(m+i_2)$ $O$-move, then the two stickers remain $(m+i_1, m+i_2, j)$-paired in $C'$.
\end{itemize}

Step 3 of the proof uses a counting argument to significantly restrict the possible moves in $m_1, \ldots, m_{k'}$. In particular, consider the following classification of moves into disjoint types:
\begin{itemize}
\item ``$O$-moves'': index-$(m+i)$ moves with $i \in O$
\item ``$T$-moves'': index-$(m+i)$ moves with $i \in T$
\item ``$M$-moves'': index-$(m+i)$ moves with $i \in M$
\item ``$J$-moves'': index-$j$ moves with $j \in J = \{1, \ldots, m\}$
\item ``vertical face moves'': face moves of faces $+x$, $+y$, $-x$, or $-y$
\item ``other moves'': all other moves
\end{itemize}

We show using the results from Steps 1 and 2 that there must be a $J$-move or two vertical face moves between each pair of $O$-moves in $m_1, \ldots, m_{k'}$. As a result, we can count the number of moves of each type as follows:

Let $c_{O}$, $c_{T}$, $c_{M}$, $c_{vertical}$, $c_{J}$, and $c_{other}$ be the number of moves of each type. We derive the following constraints:
\begin{itemize}
\item $c_{O} = |O|$
\item $c_{T} = 2|T|$
\item $c_{M} \ge 3|M|$
\item $c_{J} + \frac{1}{2}c_{vertical} \ge |O| - 1$
\item $c_{other} \ge 0$
\end{itemize}

Adding these together, we find that 
$$k' - \frac{1}{2}c_{vertical} = c_{O} + c_{T} + c_{M} + \frac{1}{2}c_{vertical} + c_{J} + c_{other} \ge |O| + 2|T| + 3|M| + (|O| - 1) = k + |M|.$$

The above shows that $k' \ge k$, but we also know that $k' \le k$. Thus, equality must hold at each step. Working out the details, we find that $c_{O} = |O|$, $c_{T} = 2|T|$, $c_{J} = |O| - 1$, and $c_{M} = c_{vertical} = c_{other} = 0$. Thus, the counting argument in this step shows that the only moves in $m_1, \ldots, m_{k'}$ other than $O$-moves and $T$-moves are the $|O| - 1$ quantity of $J$-moves which are between $O$-moves.

In Step 4, we further restrict the possibilities. In particular, we show the following:
\begin{itemize}
\item Since there are no face moves, the index-$(m+i)$ $O$-move for $i \in O$ can only be a counterclockwise $z$ turn of slice $(m+i)$. Similarly the index-$(m+i)$ $T$-moves for $i \in T$ are a clockwise $z$ turn and a $z$ flip of slice $(m+i)$.
\item Consider the elements $i \in O$ in the order in which their $O$-moves occur. We show that if $i_1$ is immediately before $i_2$ in this order, then it must be the case that $l_{i_1}$ differs from $l_{i_2}$ in exactly one bit.
\item Furthermore, the one $J$-move between two consecutive $O$-moves of slices $m+i_1$ and $m+i_2$ must rotate the $x$ slice whose index is the unique index $j$ at which strings $l_{i_1}$ and $l_{i_2}$ differ.
\end{itemize}

At this point, we are almost done. Consider the elements $i \in O$ in the order in which their $O$-moves occur. The corresponding bitstring $l_i$ in the same order have the property that each $l_i$ is at Hamming distance one from the next. In Step 5, we use the ideas of paired stickers to show that $T$ is empty, and as a result conclude that $O = \{1, \ldots, n\}$ and therefore that the above ordering of the $l_i$s is an ordering of all the $l_i$s in which each $l_i$ has Hamming distance one from the next. In other words, we show our desired result: that $l_1, \ldots, l_n$ is a ``yes'' instance to the Promise Cubical Hamiltonian Path problem.

\subsection{Step 1: restricting the set of possible index-$(m+i)$ moves}

As stated in the proof outline, we will prove the following list of results in this section

\begin{itemize}
\item $Z$ is empty.
\item If $i \in O$, then the sole index-$(m+i)$ move must be a counterclockwise $z$ turn.
\item If $i \in T$, then the two index-$(m+i)$ moves must be a clockwise $z$ turn and a $z$ flip in some order.
\item If $i \in O \cup T$, then any move of $z$ slice $(m+i)$ must occur at a time when faces $+x$, $+y$, $-x$, and $-y$ all have zero rotation and any move of $z$ slice $-(m+i)$ must occur at a time when these faces all have rotation $180^\circ$.
\end{itemize}

We begin with a preliminary result concerning the coloring of the solved configuration $C' = (m_{k'} \circ \cdots \circ m_1)(C_t)$. 

\begin{lemma}
The solved Rubik's Cube configuration $C'$ has the same face colors as $C_0$.
\end{lemma}

\begin{proof}
Consider the sticker with both coordinates $u$ on any face of $C_0$. No index-$u$ moves occur within $m_{k'}\circ \cdots \circ m_1$ by definition of $u$. No index-$u$ moves occur within $t = a_1 \circ b_1 \circ \cdots \circ b_n$ because $t$ is defined entirely using moves of slices whose indices have absolute values at most $m+n$ and $u > m+n$. As a result, the sticker in question is never moved off of the face it starts on by the transformation $m_{k'}\circ \cdots \circ m_1 \circ t$. Applying transformation $m_{k'}\circ \cdots \circ m_1 \circ t$ to $C_0$ yields $C'$, so the sticker is on the same face in $C'$ as it is in $C_0$. Since both $C_0$ and $C'$ are solved configurations, we conclude that configuration $C'$ has the same face colors as $C_0$.
\end{proof}

Using this, we can show the first desired result:

\begin{lemma}
$Z$ is empty.
\end{lemma}

\begin{proof}
Suppose for the sake of contradiction that $m_1, \ldots, m_{k'}$ contains no index-$(m+i)$ move for some $i \in \{1, \ldots, n\}$.

Then consider the sticker with coordinates $(x, z) = (u, m+i)$ on the $+y$ face of $C_b$. Configuration $C'$ can be obtained from configuration $C_b$ by transformation $m_{k'} \circ \cdots \circ m_1 \circ a_1$. We know, however, that moves $m_1, \ldots, m_{k'}$ include no index-$(m+i)$ or index-$u$ moves. Similarly, since $a_1 = (x_1)^{(l_1)_1} \circ \cdots \circ (x_m)^{(l_1)_m}$, we see that $a_1$ consists of no index-$j$ moves with $j > m$. Since both $m+i$ and $u$ are greater than $m$, we can conclude that transformation $m_{k'} \circ \cdots \circ m_1 \circ a_1$ can be built without any index-$(m+i)$ or index-$u$ moves. As a result, this transformation does not move the sticker in question to a different face.

We then see that the sticker with coordinates $(x, z) = (u, m+i)$ on the $+y$ face of $C_b$ is also on the $+y$ face of $C'$. By Theorem~\ref{thm:cube_coloring}, we see that the color of this sticker in $C_b$ is red. However, the $+y$ face of $C'$ is supposed to be the same color as the $+y$ face of $C_0$: green. 

By contradiction, we see as desired that $m_1, \ldots, m_{k'}$ must contain some index-$(m+i)$ move for all $i \in \{1, \ldots, n\}$
\end{proof}

The rest of what we wish to show concerns index-$(m+i)$ moves where $i \in O \cup T$. 

For any $i$, we can restrict our attention to a specific set of stickers as in the following definition:

\begin{definition}
Define the \emph{special stickers} to be the $48$ stickers in $C_b$ with coordinates $\pm u$ and $\pm (m+i)$ (eight special stickers per face). 
\end{definition}

Notice that by Theorem~\ref{thm:cube_coloring}, all but $8$ special stickers have the same color as the color of their starting face in $C_0$. This motivates the following further definition:

\begin{definition}
Define the \emph{correctly placed stickers} to be the $40$ special stickers which have the same color as their starting face has in $C_0$. Define the \emph{misplaced stickers} to be the other $8$ special stickers. 
\end{definition}

Of the $8$ misplaced stickers, the two on the $+y$ face have the color of the $-x$ face in $C_0$, the two on the $-x$ face have the color of the $-y$ face in $C_0$, the two on the $-y$ face have the color of the $+x$ face in $C_0$, and the two on the $+x$ face have the color of the $+y$ face in $C_0$. In short, starting at $C_b$, the $8$ misplaced stickers must each move one face counterclockwise around the $z$ axis in order to end up on the face whose color in $C_0$ matches the color of the sticker. 

Next, consider the effect that move sequence $m_1, \ldots, m_{k'}$ must have on the special stickers

\begin{lemma}
When starting in configuration $C_b$, move sequence $m_1, \ldots, m_{k'}$ must move the misplaced stickers one face counterclockwise around the $z$ axis and must return each of the correctly placed stickers to the face that sticker started on.
\end{lemma}

\begin{proof}
Configuration $C'$, which has the same coloring scheme as configuration $C_0$, can be reached from configuration $C_b$ by applying transformation $m_{k'} \circ \cdots \circ m_1 \circ a_1$. Therefore, the $8$ misplaced stickers must be moved counterclockwise one face around the $z$ axis and the $40$ correctly placed stickers must stay on the same face due to this transformation. Notice that the only moves which transfer special stickers between faces are index-$u$ and index-$(m+i)$ moves. The only other moves that even affect special stickers are face moves. As previously argued, $a_1$ contains no index-$u$ moves. In fact, $a_1$ does not contain face moves or index-$(m+i)$ moves either and so $a_1$ does not move any of the special stickers at all. 

In other words, the effect of this transformation ($m_{k'} \circ \cdots \circ m_1 \circ a_1$) on the special stickers is the same as the effect of just the transformation $m_1, \ldots, m_{k'}$. Thus, $m_1, \ldots, m_{k'}$ must move the misplaced stickers one face counterclockwise around the $z$ axis and must return each of the correctly placed stickers to the face that sticker started on. 
\end{proof}

This allows us to directly prove the next two parts of our desired result:

\begin{lemma}
\label{lemma:cube_O_move_types}
If $i \in O$, then the sole index-$(m+i)$ move must be a counterclockwise $z$ turn.
\end{lemma}

\begin{proof}
Consider the result of move sequence $m_1, \ldots, m_{k'}$ when starting in configuration $C_b$. We showed above that the $8$ misplaced stickers must each move one face counterclockwise around the $z$ axis and the correctly placed stickers must stay on the same face. Furthermore, the only moves which cause special stickers to change faces are index-$u$ or index-$(m+i)$ moves. Since $m_1, \ldots, m_{k'}$ includes no index-$u$ moves and includes exactly one index-$(m+i)$ move (for $i \in O$), we see that the special stickers only change faces during the sole index-$(m+i)$ move in $m_1, \ldots, m_{k'}$. 

Every slice with index $\pm (m+i)$ contains exactly $8$ special stickers; therefore the sole index-$(m+i)$ move must cause exactly $8$ of the special stickers to change faces. 

In order for the $8$ misplaced stickers to change faces and for the correctly placed stickers not to, it must be the case that the single index-$(m+i)$ move relocates exactly the $8$ misplaced stickers. These stickers are on the $\pm x$ and $\pm y$ faces. Since the single index-$(m+i)$ move affects $8$ stickers on the $\pm x$ and $\pm y$ faces and moves each of these stickers exactly one face counterclockwise around the $z$ axis, it must be the case that this move is a counterclockwise $z$ slice turn. As desired, the sole index-$(m+i)$ move is a counterclockwise $z$ turn.
\end{proof}

\begin{lemma}
\label{lemma:cube_T_move_types}
If $i \in T$, then the two index-$(m+i)$ moves must be a clockwise $z$ turn and a $z$ flip in some order.
\end{lemma}

\begin{proof}
As before, consider the result of move sequence $m_1, \ldots, m_{k'}$ when starting in configuration $C_b$. The $8$ misplaced stickers must each move one face counterclockwise around the $z$ axis and the correctly placed stickers must stay on the same face. Since the only moves which cause special stickers to change faces are index-$u$ or index-$(m+i)$ moves, the only moves among $m_1, \ldots, m_{k'}$ which move special stickers between faces are the two index-$(m+i)$ moves (for $i \in T$).

Notice that every slice with index $\pm (m+i)$ contains exactly $8$ special stickers, so each of the two index-$(m+i)$ moves must cause exactly $8$ of the special stickers to change faces. 

We proceed by casework: 
\begin{itemize}
\item If exactly one of the two index-$(m+i)$ moves is an $x$ or $y$ move, then at least one of the correctly placed stickers from the $+z$ face is moved from that face and never returned there. Note that correctly placed stickers are supposed to end up on their starting faces.
\item If both index-$(m+i)$ moves are $x$ moves, then the misplaced stickers from the $+y$ face never leave that face. Note that misplaced stickers are supposed to move from their starting faces.
\item If both index index-$(m+i)$ moves are $y$ moves, then the misplaced stickers from the $+x$ face never leave that face. Note that misplaced stickers are supposed to move from their starting faces.
\item If the first index-$(m+i)$ move is an $x$ move and the second is a $y$ move, then each of the misplaced stickers from the $+x$ face end up on the $\pm x$ or $\pm z$ faces. Note that misplaced stickers from the $+x$ face are supposed to move to the $+y$ face.
\item If the first index-$(m+i)$ move is a $y$ move and the second is an $x$ move, then each of the misplaced stickers from the $+y$ face end up on the $\pm y$ or $\pm z$ faces. Note that misplaced stickers from the $+y$ face are supposed to move to the $-x$ face.
\end{itemize}
Since all these cases lead to contradiction, we can conclude that the only remaining case holds: both index-$(m+i)$ moves must be $z$ moves.

Next suppose for the sake of contradiction that the $8$ special stickers which are moved by one index-$(m+i)$ move are not the same as the special stickers moved by the other index-$(m+i)$ move. Any special sticker moved by exactly one of these moves will change faces and must therefore be a misplaced sticker. That sticker must move one face counterclockwise around the $z$ axis. Since each of the two index-$(m+i)$ moves includes at least one sticker that is not moved by the other index-$(m+i)$ move we can conclude that the two index-$(m+i)$ moves are both counterclockwise $z$ turns. Then any sticker moved by both index-$(m+i)$ moves is moved two faces counterclockwise around the $z$ axis. This is not the desired behavior for any of the special stickers so none of the stickers can be moved by both index-$(m+i)$ moves. Thus there are a total of $16$ different special stickers, each of which is moved by exactly one of the two index-$(m+i)$ moves. All $16$ of these stickers end up on a different face from the one they started at. This is a contradiction since there are only $8$ misplaced stickers.

We conclude that the two moves affect the same $8$ stickers. The only way to rotate a total of one quarter rotation counterclockwise with two moves is using one clockwise turn and one flip. Thus, as desired, the two index-$(m+i)$ moves for $i \in T$ must be a clockwise $z$ turn and a $z$ flip in some order.
\end{proof}

Finally, we have only one thing left to prove in this section:

\begin{lemma}
If $i \in O \cup T$, then any move of $z$ slice $(m+i)$ must occur at a time when faces $+x$, $+y$, $-x$, and $-y$ all have zero rotation and any move of $z$ slice $-(m+i)$ must occur at a time when faces $+x$, $+y$, $-x$, and $-y$ all have rotation $180^\circ$. 
\end{lemma}

\begin{proof}
As before, consider the result of move sequence $m_1, \ldots, m_{k'}$ when starting in configuration $C_b$. The $8$ misplaced stickers must each move one face counterclockwise around the $z$ axis and the correctly placed stickers must stay on the same face. The only moves which cause special stickers to change faces are index-$u$ or index-$(m+i)$ moves, though face moves also move special stickers. The only moves among $m_1, \ldots, m_{k'}$ which move special stickers between faces are the one or two index-$(m+i)$ moves (for $i \in O \cup T$). Furthermore, as shown in the proofs of Lemmas~\ref{lemma:cube_O_move_types} and~\ref{lemma:cube_T_move_types}, the special stickers which are affected by these moves are exactly the misplaced stickers. In other words, throughout the entire move sequence $m_1, \ldots, m_{k'}$, the only moves which affect the correctly placed stickers are the face moves.

Let $m_j$ be any index-$(m+i)$ move. Note that according to Lemmas~\ref{lemma:cube_O_move_types} and~\ref{lemma:cube_T_move_types}, $m_j$ rotates a $z$ slice.

Consider the six correctly placed stickers on one of the $\pm x$ or $\pm y$ faces. Since these stickers are only ever affected by face moves, their coordinates within the face are completely determined by the total rotation of the face so far. If the total rotation so far is $0$, then the six correctly placed stickers are in the positions with coordinates $\pm u$ and $\pm (m+i)$ and with $z \ne (m+i)$. If the total rotation so far is $90^\circ$, then the six correctly placed stickers are in the positions with coordinates $\pm u$ and $\pm (m+i)$ and with $x \ne -(m+i)$ for the $\pm y$ faces or $y \ne (m+i)$ for the $\pm x$ faces. If the total rotation so far is $180^\circ$, then the six correctly placed stickers are in the positions with coordinates $\pm u$ and $\pm (m+i)$ and with $z \ne -(m+i)$. If the total rotation so far is $270^\circ$, then the six correctly placed stickers are in the positions with coordinates $\pm u$ and $\pm (m+i)$ and with $x \ne (m+i)$ for the $\pm y$ faces or $y \ne -(m+i)$ for the $\pm x$ faces. 

The only way for move $m_j$ to avoid affecting these stickers if $m_j$ rotates $z$ slice $(m+i)$ is for the stickers to be in the positions with $z \ne (m+i)$. In other words, the total rotation of the face must be $0$. The only way for move $m_j$ to avoid affecting these stickers if $m_j$ rotates $z$ slice $-(m+i)$ is for the stickers to be in the positions with $z \ne -(m+i)$. In other words, the total rotation of the face must be $180^\circ$. Note that this logic applies to each of the $\pm x$ and $\pm y$ faces. In other words, if $m_j$ is some move with index $(m+i)$, then each of the $\pm x$ and $\pm y$ faces must have rotation $0$ and if $m_j$ is some move with index $-(m+i)$, then each of the $\pm x$ and $\pm y$ faces must have rotation $180^\circ$.
\end{proof}

\subsection{Step 2: exploring properties of paired stickers}

As stated in the proof outline, this step of the proof explores the properties of paired stickers.

\begin{lemma}
\label{lemma:cube_simple_pairing_result}
If two stickers are $(p_1, p_2, q)$-paired, then they remain $(p_1, p_2, q)$-paired after one move unless the move is an index-$p_1$ move or an index-$p_2$ move which moves one of the stickers.
\end{lemma}

\begin{proof}
Consider the effect of any move on the two stickers. 

If the move doesn't affect either sticker, then the two stickers maintain their coordinates (and therefore also stay on the same face quadrant and slice). Thus the two stickers remain $(p_1, p_2, q)$-paired.

If the move moves both stickers, then they both rotate the same amount. In other words, as far as those two stickers are concerned, the effect of the move is the same as the effect of rotating the entire Rubik's Cube. When rotating the Rubik's Cube, two stickers sharing a slice continue to share a slice, two stickers sharing a face quadrant continue to share a face quadrant, and each sticker maintains the same set of coordinate absolute values as it had before. Thus, the two stickers remain $(p_1, p_2, q)$-paired.

Clearly, the only way for the stickers to no longer be $(p_1, p_2, q)$-paired is for the move to affect exactly one of the stickers. The possible moves affecting the stickers are face moves, index-$q$ moves, index-$p_1$ moves, and index-$p_2$ moves. Among these, face moves and index-$q$ moves necessarily affect either both stickers in the pair or neither. Thus, the only way for the stickers to stop being $(p_1, p_2, q)$-paired is via an index-$p_1$ move or an index-$p_2$ move which moves one of the stickers.
\end{proof}

Using this, we prove the following lemma:

\begin{lemma}
\label{lemma:cube_complex_pairing_result}
Suppose $i_1, i_2 \in O$ and $j \in \{1, 2, \ldots, m\}$. Then consider any pair of stickers that are $(m+i_1, m+i_2, j)$-paired in $C_b$. If there are no face moves of faces $+x$, $+y$, $-x$, and $-y$ and no index-$j$ moves that affect either of the stickers between the index-$(m+i_1)$ $O$-move and the index-$(m+i_2)$ $O$-move, then the two stickers remain $(m+i_1, m+i_2, j)$-paired in $C'$.
\end{lemma}

\begin{proof}
Consider two $(m+i_1, m+i_2, j)$-paired stickers in $C_b$. Suppose that there exists neither an index-$j$ move affecting one of the stickers nor a face move of face $+x$, $+y$, $-x$, or $-y$ between the index-$(m+i_1)$ and index-$(m+i_2)$ $O$-moves. Let $m_\alpha$ be the index-$(m+i_1)$ $O$-move and let $m_\beta$ be the index-$(m+i_2)$ $O$-move. Without loss of generality, suppose $m_\alpha$ occurs before $m_\beta$.

Since there are no $+x$, $+y$, $-x$, or $-y$ face moves between $m_\alpha$ and $m_\beta$, we know that the rotations of these faces remain the same at the times of both moves. Applying the results from Step 1, either $m_\alpha$ and $m_\beta$ are both counterclockwise turns of $z$ slices $(m+i_1)$ and $(m+i_2)$ or $m_\alpha$ and $m_\beta$ are counterclockwise turns of $z$ slices $-(m+i_1)$ and $-(m+i_2)$.

Configuration $C'$ can be obtained from configuration $C_b$ by applying transformation $m_{k'} \circ \cdots \circ m_2 \circ m_1 \circ a_1$. Since $a_1$ consists of some number of $x$-slice turns, we can represent this transformation as a sequence of moves. We know that since the stickers are $(m+i_1, m+i_2, j)$-paired in $C_b$, they must remain $(m+i_1, m+i_2, j)$-paired until immediately before the first index-$(m+i_1)$ or index-$(m+i_2)$ move: $m_\alpha$. We will show below that because of our assumption, the stickers will also end up $(m+i_1, m+i_2, j)$-paired immediately after $m_\beta$ in all cases.

The first case is that the stickers are on face $+z$ or face $-z$ immediately before $m_\alpha$. In that case, move $m_\alpha$, which is a $z$ move, will not affect either sticker. As a result, the two stickers will remain $(m+i_1, m+i_2, j)$-paired after $m_\alpha$. With the exception of $m_\alpha$ and $m_\beta$, the only moves in $m_1, \ldots, m_{k'}$ which move these two stickers between faces are index-$j$ moves. But by assumption, there are no index-$j$ moves occurring between $m_\alpha$ and $m_\beta$ which affect the stickers. Thus, immediately before $m_\beta$, the two stickers will still be $(m+i_1, m+i_2, j)$-paired and will still be on face $+z$ or face $-z$. As a result, $m_\beta$ will also not affect the stickers. Therefore, they will remain $(m+i_1, m+i_2, j)$-paired immediately after $m_\beta$.

The second case is that the stickers are on face $+x$, $+y$, $-x$, or $-y$ immediately before $m_\alpha$. Between $m_\alpha$ and $m_\beta$, the only moves in $m_1, \ldots, m_{k'}$ which move these two stickers are face moves and index-$j$ moves. No matter how $m_\alpha$ affects the two stickers, they will both remain on the four faces $+x$, $+y$, $-x$, and $-y$. By assumption, neither sticker will be moved by an index-$j$ move until $m_\beta$. Since the stickers are on faces $+x$, $+y$, $-x$, or $-y$, the assumption tells us that neither sticker will be moved by a face move until $m_\beta$ either. Thus, the next move after $m_\alpha$ which affects either sticker is $m_\beta$. 

Note that immediately before $m_\alpha$, the first sticker has $z$ coordinate $(m+i_1)$ if and only if the second sticker has $z$ coordinate $(m+i_2)$. Similarly, the first sticker has $z$ coordinate $-(m+i_1)$ if and only if the second sticker has $z$ coordinate $-(m+i_2)$. This is simply a consequence of the definition of $(m+i_1, m+i_2, j)$-paired stickers. We know that $m_\alpha$ and $m_\beta$ together either rotate $z$ slices $(m+i_1)$ and $(m+i_2)$ counterclockwise one turn or rotate $z$ slices $-(m+i_1)$ and $-(m+i_2)$ counterclockwise one turn. Thus in any case we see that over the course of the moves from $m_\alpha$ to $m_\beta$, either both stickers are rotated counterclockwise one turn around the $z$ axis or neither is. As far as the two stickers are concerned, that is equivalent to a rotation of the entire Rubik's Cube. That means that in this case as well, the two stickers remain $(m+i_1, m+i_2, j)$-paired immediately after $m_\beta$.

We see that in both cases the two stickers remain $(m+i_1, m+i_2, j)$-paired immediately after $m_\beta$. Since there are no index-$(m+i_1)$ or index-$(m+i_2)$ moves after $m_\beta$, we know that the two stickers will continue to be $(m+i_1, m+i_2, j)$-paired until $C'$. 
\end{proof}

\subsection{Step 3: classifying possible moves with a counting argument}

As stated in the proof outline, this step uses a counting argument to restrict the possible moves in $m_1, \ldots, m_{k'}$. 

To begin, we show the following:

\begin{lemma}
There must be a $J$-move or two vertical face moves between each pair of $O$-moves in $m_1, \ldots, m_{k'}$. 
\end{lemma}

\begin{proof}
Consider any pair of $O$-moves $m_\alpha$ and $m_\beta$ which occur in that order. Suppose $m_\alpha$ is an index-$(m+i_1)$ move and $m_\beta$ is an index-$(m+i_2)$ move. Let $j$ be an index such that $(l_{i_1})_j$ differs from $(l_{i_2})_j$. 

Notice that the $(m+i_1, m+i_2, j)$-paired stickers on face $+y$ with $(x,z)$ coordinates $(j, m+i_1)$ and $(j, m+i_2)$ have different colors in $C_b$ (see Theorem~\ref{thm:cube_coloring}). Therefore they cannot be $(m+i_1, m+i_2, j)$-paired in $C'$. By the contraposative of Lemma~\ref{lemma:cube_complex_pairing_result}, we see that at least one index-$j$ move affecting one of these stickers or at least one face move of faces $+x$, $+y$, $-x$, or $-y$ must occur between $m_\alpha$ and $m_\beta$.

We know from the results of Step 1 that at the times of $m_\alpha$ and $m_\beta$, the four faces $\pm x$ and $\pm y$ must either each have total rotation of $0$ or total rotation of $180^\circ$. Thus between the two moves, either the rotations of all four faces must change, or the rotation of any face that changes must also change back. Therefore it is impossible for exactly one face move of faces $+x$, $+y$, $-x$, and $-y$ to occur between these two moves. 

In other words, we have shown that at least one $J$-move or at least two vertical face moves must occur between $m_\alpha$ and $m_\beta$.
\end{proof}

\begin{corollary}
\label{corollary:cube_J_move_between}
If 
\begin{itemize}
\item $m_\alpha$ and $m_\beta$ are index-$(m+i_1)$ and index-$(m+i_2)$  $O$-moves,
\item $l_{i_1}$ and $l_{i_2}$ differ in bit $j$, and 
\item there are no vertical face moves between $m_\alpha$ and $m_\beta$,
\end{itemize}
then there must be an index-$j$ $J$-move between $m_\alpha$ and $m_\beta$.
\end{corollary}

\begin{proof}
This follows directly from one of the cases in the previous proof.
\end{proof}

With that done, we can count the number of moves of each type as follows:

There is exactly one $O$-move for each $i \in O$ (the sole index-$(m+i)$ move), so therefore $c_{O} = |O|$. 

There are exactly two $T$-move for each $i \in T$ (the two index-$(m+i)$ moves), so therefore $c_{T} = 2|T|$. 

There are at least three $M$-moves for each $i \in M$ (the index-$(m+i)$ moves), so therefore $c_{M} \ge 3|M|$. 

Consider the $O$-moves in order. Between the $c_{O} = |O|$ different $O$-moves there are $|O| - 1$ gaps. As shown above, each such gap must contain either at least one $J$-move or at least two vertical face moves. Therefore the number of $J$-moves plus half the number of vertical face moves upper-bounds the number of gaps: $c_{J} + \frac{1}{2}c_{vertical} \ge |O| - 1$.

Finally, $c_{other} \ge 0$.

Putting this together, we see the following:
\begin{align*}
k' &= c_{O} + c_{T} + c_{M} + c_{vertical} + c_{J} + c_{other} \\
&= c_{O} + c_{T} + c_{M} + \left(c_{J} + \frac{1}{2}c_{vertical}\right) + \frac{1}{2}c_{vertical} + c_{other} \\
&\ge |O| + 2|T| + 3|M| + (|O| - 1) + \frac{1}{2}c_{vertical} \\
&= 2|O| + 2|T| + 3|M| - 1 + \frac{1}{2}c_{vertical} \\
&= 2n - 1 + |M| + \frac{1}{2}c_{vertical} \\
&= k + |M| + \frac{1}{2}c_{vertical} \\
&\ge k
\end{align*}

The above shows that $k' \ge k$, but we also know that $k' \le k$. Thus, equality must hold at each step. In particular, $c_{M}$ must equal $3|M|$, $\left(c_{J} + \frac{1}{2}c_{vertical}\right)$ must equal $|O| - 1$, $c_{other}$ must equal $0$, and $|M| + \frac{1}{2}c_{vertical}$ must equal $0$.

Since $|M| + \frac{1}{2}c_{vertical} = 0$, we can conclude that both $|M|$ and $c_{vertical}$ are equal to $0$. Thus $c_{M} = 3|M| = 0$ also holds. All together, this shows that $c_{O} = |O|$, $c_{T} = 2|T|$, $c_{J} = |O| - 1$, and $c_{M} = c_{vertical} = c_{other} = 0$.

\subsection{Step 4: further restricting possible move types}

As stated in the proof outline, we will prove the following list of results in this section:

\begin{itemize}
\item Since there are no face moves, the index-$(m+i)$ $O$-move for $i \in O$ can only be a counterclockwise $z$ turn of slice $(m+i)$. Similarly the index-$(m+i)$ $T$-moves for $i \in T$ are a clockwise $z$ turn and a $z$ flip of slice $(m+i)$.
\item Consider the elements $i \in O$ in the order in which their $O$-moves occur. We show that if $i_1$ is immediately before $i_2$ in this order, then it must be the case that $l_{i_1}$ differs from $l_{i_2}$ in exactly one bit.
\item Furthermore, the one $J$-move between two consecutive $O$-moves of slices $m+i_1$ and $m+i_2$ must rotate the $x$ slice whose index is the unique index $j$ at which strings $l_{i_1}$ and $l_{i_2}$ differ.
\end{itemize}

\begin{lemma}
If $i \in O$, the single index-$(m+i)$ move in $m_1, \ldots, m_{k'}$ is a counterclockwise $z$ turn of slice $(m+i)$.
\end{lemma}

\begin{proof}
We have already seen that the move in question must be either a counterclockwise $z$ turn of slice $(m+i)$ or a counterclockwise $z$ turn of slice $-(m+i)$. Furthermore, the slice being rotated is slice $(m+i)$ if at the time of the move each vertical face ($\pm x$ and $\pm y$) has the total rotation $0$. We have already seen, however, that none of the moves in $m_1, \ldots, m_{k'}$ are face moves. Thus the total rotation of each face is always $0$, and as desired, the move in question is a counterclockwise $z$ turn of slice $(m+i)$.
\end{proof}

\begin{lemma}
If $i \in T$, the two index-$(m+i)$ moves in $m_1, \ldots, m_{k'}$ are a clockwise $z$ turn of slice $(m+i)$ and a $z$ flip of slice $(m+i)$.
\end{lemma}

\begin{proof}
This proof follows analagously to the previous.
\end{proof}

\begin{lemma}
Suppose that $m_\alpha$ and $m_\beta$ are two $O$-moves of slices $(m+i_1)$ and $(m+i_2)$ with no other $O$-moves between them. It must be the case that $l_{i_1}$ differs from $l_{i_2}$ in exactly one bit.
\end{lemma}

\begin{proof}
We have already seen that there must be at least one $J$-move between $m_\alpha$ and $m_\beta$. In fact, there has to be exactly one $J$-move in each of the $|O| - 1$ ``gaps'' between $O$-moves, so there can only be one $J$-move between $m_\alpha$ and $m_\beta$.

We saw in Corollary~\ref{corollary:cube_J_move_between}, however, that if $l_{i_1}$ and $l_{i_2}$ differ in bit $j$, then there must be an index-$j$ $J$-move between $m_\alpha$ and $m_\beta$. As desired, we conclude that $l_{i_1}$ and $l_{i_2}$ must differ in at most one bit $j$. Since the bitstrings are all distinct, this is exactly what we were trying to show.
\end{proof}

\begin{lemma}
Suppose that $m_\alpha$ and $m_\beta$ are two $O$-moves of slices $(m+i_1)$ and $(m+i_2)$ with no other $O$-moves between them. If $l_{i_1}$ differs from $l_{i_2}$ in bit $j$, then it must be the case that the one $J$-move between $m_\alpha$ and $m_\beta$ must rotate the $x$ slice with index $j$.
\end{lemma}

\begin{proof}
We know that the $J$-move in question must rotate a slice with index $\pm j$. We want to show that the move rotates the $x$ slice with index $j$ in particular.

Consider the pair of stickers in $C_b$ at $(x, z)$ coordinates $(j, m+i_1)$ and $(j, m+i_2)$ on the $+y$ face and also the pair of stickers in $C_b$ at $(x, y)$ coordinates $(j, -(m+i_1))$ and $(j, -(m+i_2))$ on the $+z$ face. These two pairs of stickers are both $(m+i_1, m+i_2, j)$-paired. Furthermore, each of these two pairs contain stickers of two different colors (see Theorem~\ref{thm:cube_coloring}). 

To transition from $C_b$ to $C'$, we apply transformation $m_k \circ \cdots \circ m_1 \circ a_1$. In other words, we apply a sequence of moves starting with some number of $x$ turns (making up $a_1$) and then proceeding through move sequence $m_1, \ldots, m_k$. Because the solution contains no face moves, the only moves in this list before $m_\alpha$ which affect the four stickers in question are rotations of the $x$ slice with index $j$. No matter how much or how little this slice rotates, one of the two pairs of stickers will be on face $+z$ or $-z$.

Consider that pair. Move $m_\alpha$ will be a counterclockwise $z$ turn and therefore will not affect either sticker in the pair. That pair of stickers cannot be $(m+i_1, m+i_2, j)$-paired in $C'$ since they have different colors. Since $m_\beta$ is the only other index-$(m+i_1)$ or index-$(m+i_2)$ move, we can conclude from Lemma~\ref{lemma:cube_simple_pairing_result} that one of the two stickers must be affected by $m_\beta$. In order for that to be the case, however, the sole $J$-move between $m_\alpha$ and $m_\beta$ must move the stickers in this pair off of the $\pm z$ face. Notice that the $J$-move between $m_\alpha$ and $m_\beta$ must rotate a slice with index $\pm j$. Since there are no face moves in the solution, the only option which meets the requirements is to have the $J$-move rotate the $x$ slice with index $j$.
\end{proof}

\subsection{Step 5: showing $T$ is empty}

As stated in the proof outline, the purpose of this step is to show that $T$ is empty. That on its own is sufficient to complete the proof.

\begin{lemma}
\label{lemma:cube_O_move_pattern}
When applying the move sequence $a_1, m_1, \ldots, m_k$ to $C_b$, the stickers with $z = i$ and $1 \le x \le n$ of face $+y$ for $i \in O$ immediately after the $O$-move of slice $(m + i)$ are the ones which started in the corresponding positions $z =i$ and $1 \le -y \le n$ of the face $+x$ in $C_b$.
\end{lemma}

\begin{proof}
Let $m_\alpha$ be the $O$-move of slice $(m + i)$.

Consider the stickers in positions $z = i$ and $1 \le x \le n$ of face $+y$ for $i \in O$ immediately after the move $m_\alpha$. These stickers were moved there by $m_\alpha$ from positions $z =i$ and $1 \le -y \le n$ of the face $+x$. 

All $O$-moves and $T$-moves prior to $m_\alpha$ affect $z$ slices whose indices are not $i$. All $J$-moves and all moves comprising $a_1$ affect non-face $x$ slices and therefore don't affect the $+x$ face. As a result, no move in $a_1, m_1, \ldots, m_k$ before $m_\alpha$ affects the stickers with $z = i$ and $-n \le y \le -1$ on the $+x$ face. Thus, the stickers in positions $z =i$ and $1 \le -y \le n$ of the face $+x$ immediately before $m_\alpha$ are the same as the stickers in those positions in configuration $C_b$.

As desired, the stickers with $z = i$ and $1 \le x \le n$ of face $+y$ for $i \in O$ immediately after the move $m_\alpha$ are the ones which started in the corresponding positions $z =i$ and $1 \le -y \le n$ of the face $+x$ in $C_b$.
\end{proof}

\begin{lemma}
\label{lemma:cube_T_move_pattern}
When applying the move sequence $a_1, m_1, \ldots, m_k$ to $C_b$, the stickers with $z = i$ and $1 \le x \le n$ of face $+y$ for $i \in T$ after the second $T$-move rotating a slice with index $(m+i)$ are the ones which started in the corresponding positions $z =i$ and $1 \le -y \le n$ of the face $+x$ in $C_b$.
\end{lemma}

\begin{proof}
Let $m_\alpha$ and $m_\beta$ be the two $T$-moves of slice $(m + i)$.

Consider the stickers in positions $z = i$ and $1 \le x \le n$ of face $+y$ immediately after $m_\beta$. These stickers were moved there by $m_\beta$ either from positions $z =i$ and $1 \le y \le n$ of the face $-x$ or from positions $z = i$ and $1 \le -x \le n$ of face $-y$ (depending on whether the second $T$-move is the turn or the flip). 

In either case, none of the moves between $m_\alpha$ and $m_\beta$ could have affected any of these stickers (since the moves in that interval are all either $O$- or $T$- moves moving $z$ slices of other indices or $J$-moves moving $x$ slices with indices $1$ through $n$). Therefore immediately before $m_\alpha$, these stickers were in positions $z =i$ and $1 \le -y \le n$ of the face $+x$. Once again, no moves before that could affect these stickers, so these stickers must have started in that position in $C_b$.

As desired, the stickers with $z = i$ and $1 \le x \le n$ of face $+y$ immediately after the move $m_\beta$ are the ones which started in the corresponding positions $z =i$ and $1 \le -y \le n$ of the face $+x$ in $C_b$.
\end{proof}

\begin{theorem}
$T$ is empty.
\end{theorem}

\begin{proof}
Note that $O$ cannot be empty since then the number of $J$-moves would be $|O| - 1 = -1$. 

Suppose for the sake of contradiction that $i_1 \in T$. Consider the second $T$-move of the $z$ slice with index $(m+i_1)$ in move sequence $a_1, m_1, \ldots, m_k$. Call this move $m_\alpha$. The move $m_\alpha$ cannot be seperated from every $O$-move by $J$-moves because if that were the case, there would be two $J$-moves without an $O$-move between them (or in other words there would be two $O$-moves with at least two $J$-moves between them). Thus there must be some $O$-move $m_\beta$ of slice $(m
+i_2)$ that is not seperated from $m_\alpha$ by any $J$-move. 

Consider what happens if we apply the move sequence $a_1, m_1, \ldots, m_k$ to $C_b$ until right after both $m_\alpha$ and $m_\beta$ have occurred. Call this configuration $C_{mid}$. For every $j \in \{1, \ldots, m\}$, the stickers that are in $(x, z)$ coordinates $(j, m+i_1)$ and $(j, m+i_2)$ of face $+y$ in $C_{mid}$ are $(m+i_1, m+i_2, j)$-paired. When transitioning from $C_{mid}$ to $C'$, no index-$(m+i_1)$ or index-$(m+i_2)$ moves occur, and so these stickers are also $(m+i_1, m+i_2, j)$-paired in $C'$. Thus we conclude that the stickers in each pair are the same color. 

Therefore we have that in $C_{mid}$, the stickers on face $+y$ with $z = i_2$ and $1 \le x \le n$ have the same color scheme, call it $S$, as the stickers on face $+y$ with $z = i_1$ and $1 \le x \le n$. Before we reach the configuration $C_{mid}$, the final few moves are a sequence of $O$-moves and $T$-moves including $m_\alpha$ and $m_\beta$. Furthermore, among these $O$-moves and $T$-moves, none that occur after $m_\alpha$ affect the stickers with $z = i_1$ and none that occur after $m_\beta$ affect the stickers with $z = i_2$. Therefore the color scheme of the stickers in positions $z = i_2$ and $1 \le x \le n$ of face $+y$ immediately after $m_\beta$ is the same as $S$: the color scheme of those stickers in $C_{mid}$. Similarly, the color scheme of the stickers in positions $z = i_1$ and $1 \le x \le n$ of face $+y$ immediately after $m_\alpha$ is also $S$. Using Lemmas~\ref{lemma:cube_O_move_pattern} and~\ref{lemma:cube_T_move_pattern}, we conclude that the color scheme of the stickers in positions $z = i_2$ and $1 \le -y \le n$ of face $+x$ in configuration $C_b$ is $S$ and that the color scheme of the stickers in positions $z = i_1$ and $1 \le -y \le n$ of face $+x$ in configuration $C_b$ is also $S$. This, however, is a contradiction, since those two color schemes in $C_b$ are different for any two different $i_1$ and $i_2$ (see Theorem~\ref{thm:cube_coloring}).

We conclude that $i_1 \in T$ cannot exist, and therefore that $T$ is empty.
\end{proof}

This completes the proof of Theorem~\ref{thm:cube_second_direction} outlined in Section~\ref{section:cube_second_direction}.

\subsection{Conclusion}

Theorems~\ref{thm:cube_first_direction} and~\ref{thm:cube_second_direction} and Corollaries~\ref{corollary:cube_first_direction} and~\ref{corollary:cube_second_direction} show that the polynomial-time reductions given are answer preserving. As a result, we conclude that 

\begin{theorem}
The STM/SQTM Rubik's Cube and Group STM/SQTM Rubik's Cube problems are NP-complete.
\end{theorem}

\section{Future work}
\label{section:next_steps}

In this paper, we resolve the complexity of optimally solving Rubik's Cubes under move count metrics for which a single move rotates a single slice. It could be interesting to consider the complexity of this problem under other move count metrics.

Of particular interest are the Wide Turn Metric (WTM) and Wide Quarter Turn Metric (WQTM), in which the puzzle solver can rotate any number of contiguous layers together provided they include one of the faces. These move count metrics are the closest to how one would physically solve a real-world $n \times n \times n$ Rubik's Cube: by grabbing some of the slices in the cube (including a face) from the side and rotating those slices together. We can also consider the $1 \times n \times n$ analogue of the Rubik's Cube with WTM move count metric: this would be a Rubik's Square in which a single move flips a contiguous sequence of rows or columns including a row or column at the edge of the Square. Solving this toy model could help point us in the right direction for the WTM and WQTM Rubik's Cube problems. If even the toy model resists analysis, it could be interesting to consider this toy model with missing stickers.


\bibliography{cube}
\bibliographystyle{plain}

\end{document}